\documentclass[journal]{IEEEtran}
\ifCLASSINFOpdf
\else
\fi
\usepackage{mathptmx} 
\usepackage{times} 
\usepackage{amsmath} 
\usepackage{amssymb}  
\usepackage{algorithm}
\usepackage{algorithmic}
\usepackage{textcomp}

\usepackage{amsthm}
\usepackage{graphicx}
\usepackage{epstopdf}
\usepackage{mathrsfs}
\usepackage{color}

\providecommand{\norm}[1]{\lVert#1\rVert}
\newcommand{\IM}{\mathrm{im}}

\newcommand{\QNUM}{D}
\newcommand{\Rank}{\mathrm{rank}}
\newcommand{\SPAN}{\mathrm{span}}

\newtheorem{Definition}{Definition}
\newtheorem{Theorem}{Theorem}
\newtheorem{Lemma}{Lemma}

\newtheorem{Corollary}{Corollary}
\newtheorem{Remark}{Remark}
\newtheorem{Notation}{Notation}
\newtheorem{Example}{Example}

\newtheorem{Problem}{Problem}

\begin{document}
%
\title{Model Reduction by Moment Matching for Linear Switched Systems}
%
%
%

\author{Mert Ba\c{s}tu\u{g}$^{1,2}$, Mih\'{a}ly Petreczky$^{2}$, Rafael Wisniewski$^{1}$ and John Leth$^{1}$
	\thanks{$^{1}$Department of Electronic Systems, Automation and Control, Aalborg University, 9220 Aalborg, Denmark {\tt\small mertb@es.aau.dk} Tel: +33 646897133 Fax: +33 327712917, {\tt\small raf@es.aau.dk} Tel: +45 99408762 Fax: +45 98151739, {\tt\small jjl@es.aau.dk} Tel: +45 99407973 Fax: +45 98151739}%
	\thanks{$^{2}$Department of Computer Science and Automatic Control (UR Informatique et Automatique), \'{E}cole des Mines de Douai, 59508 Douai, France {\tt\small mihaly.petreczky@mines-douai.fr} Tel: +33 327712238 Fax: +33 327712917}%
}

\markboth{IEEE Transactions on Automatic Control,~Vol.~, No.~, October~2014}%
{Model Reduction by Moment Matching for Linear Switched Systems}
%



\maketitle

\begin{abstract}
Two moment-matching methods for model reduction of linear switched systems (LSSs) are presented. The methods are similar to the Krylov subspace methods used for moment matching for linear systems. The more general one of the two methods, is based on the so called ``nice selection'' of some vectors in the reachability or observability space of the LSS. The underlying theory is closely related to the (partial) realization theory of LSSs. In this paper, the connection of the methods to the realization theory of LSSs is provided, and algorithms are developed for the purpose of model reduction. Conditions for applicability of the methods for model reduction are stated and finally the results are illustrated on numerical examples.
\end{abstract}


%
\IEEEpeerreviewmaketitle

\section{Introduction}
%
%
%
%
\IEEEPARstart{A}{}linear switched system (abbreviated by LSS) is a model of a dynamical process whose behavior changes among a number of linear subsystems depending on a logical decision mechanism, i.e., an LSS is a concatenation of linear systems. That is, the state of the linear subsystem just before a switching instant serves as the initial state for the next active linear system. The information about which local mode operates in a specific time instant, is contained in the switching signal, which can be totally arbitrary. Hence, the switching signal serves as an external input. Linear switched systems represent the simplest class of hybrid systems, they have been studied  extensively,  see \cite{liberzon2003}, \cite{Sun:Book} for an overview.


Model reduction is the problem of approximating a dynamical system with another one of smaller complexity. ``Smaller complexity'' for LSSs can refer to ``smaller number of state variables of each local mode'' or to ``smaller number of local modes''. In this work, by complexity we mean the former, and thus by model reduction we mean the approximation of the original LSS by another one, with a smaller number of states.

\textbf{Contribution of the paper}
In this paper, first we present model reduction algorithms based on partial realization theory for LSSs \cite{petreczky}. 
The main idea is to replace the original LSS by an LSS of smaller order, such that certain Markov parameters of the two LSSs are equal. Markov parameters of an LSS are the coefficients appearing in the Taylor series expansion of its input-output map around zero. More precisely, they are the high-order partial derivatives of the zero-state and zero-input responses of the LSS with respect to the dwell times (time between two consecutive changes in the switching signal) of each operating mode. Hence, if some of the lower order derivatives of the responses of two LSSs coincide, it means that their input-output behaviors are close. 

We present two methods. The first one will preserve all the Markov parameters which correspond to high-order derivatives up to order $N$ for some integer $N$. We will call this method $N$-moment matching. This is a direct counterpart of the well-known method of moment matching for linear systems, where the  reduced order model preserves the first $N$ Markov parameters of the transfer function at hand, \cite{antoulas}. The second method preserves a certain selection (not necessarily finite) of Markov parameters. The selections we allow will be referred to as \emph{nice selections}. Intuitively, a nice selection corresponds to a choice of basis of the extended controllability (resp. observability) space of an LSS \cite{Sun:Book,MP:BigArticlePartI}. The notion of nice selections is a direct generalization of the corresponding notion for linear systems \cite{Hazewinkel1,gugercin}, and in a more restricted form it appeared in \cite{petreczkypeeters1}. The second method gives the user additional flexibility in choosing
which Markov parameters should be preserved. In turn, this allows the user to focus on those Markov parameters which are relevant for the dynamical properties one wishes to preserve. For example, by choosing certain Markov parameters, it is possible to preserve the input-output behavior of the system in a certain discrete mode or even for a sequence of modes. At the end of the paper, we will present results to this effect. 
From an algorithmic point of view, both methods represent an extension of the classical Krylov subspace  based methods. 

\textbf{Motivation}
One of the main motivations for developing model reduction methods is that the order of the controller and the computation complexity of controller synthesis increase with the number of state variables of the plant model. This curse of dimensionality can be particularly troublesome for hybrid systems. The reason is as follows: A finite-state abstraction of the plant model is acquired in many of the existing control synthesis methods \cite{TabuadaBook}, subsequently one applies discrete-event control synthesis techniques to find a discrete controller for the finite-state abstraction of the plant. Usually, the states of this abstraction are not directly measurable, only some events (transition labels) are. This means that the controller has to contain a copy of the abstracted plant model, in order to be able to estimate the state of the finite-state abstraction of the plant, \cite{TabuadaBook,Wonham3,VardKupfContr}. In addition, the complexity of the control synthesis algorithm is at best polynomial in the number of states of the finite-state abstraction \cite{TabuadaBook,Wonham3,GameBook}. The situation becomes even worse when one considers the case of partial observations, i.e., when not all events (transition labels) of the finite-state abstraction are observable. This can be caused by the nature of the problem \cite{MP:HybIoDADHS09} or by the non-determinism of the abstraction. In this case, the control synthesis algorithm can have exponential complexity, \cite{GameBook,arnold_games_2003,Wonham3}, and the number of the state of the controller can be exponential in the number of the states of the abstraction. Depending on the method used and on the application at hand, the size of the finite-state abstraction can be very large, it could even be exponential in the number of continuous states of the original hybrid model, \cite{TabuadaBook}. In such cases, synthesis or implementation of controller might become very difficult, even for hybrid system of moderate size. Clearly, model reduction algorithms could be useful for such systems.

\textbf{Related work} The possibility of model reduction by moment matching for LSSs was already hinted in \cite{petreczky}, but no details were provided, no efficient algorithm was proposed, and no numerical experiments were done. Note that a naive application of the realization algorithm of \cite{petreczky} yields an algorithm whose computational complexity is exponential. Some results of this paper have appeared in \cite{bastugACC2014}. Main contributions of this paper different from \cite{bastugACC2014} can be summarized as follows: 1) Proofs for the main theorems in \cite{bastugACC2014} are presented. 2) The model reduction framework given in \cite{bastugACC2014} is generalized with the notion of nice selections. Hence, a less conservative framework is built for model reduction of LSSs, which is useful for focusing on the approximation of specific local modes. 3) This generalized framework is used to state a theorem which can be used for matching the input output behavior of a \emph{continuous time} LSS for a \emph{certain} switching sequence, with another LSS of smaller order. In \cite{bastugCDC2014}, the moment matching framework is used for matching the input-output behavior of \emph{discrete time} LSSs with a \emph{certain set} of allowed switching sequences. With respect to \cite{bastugCDC2014}, the main differences are that this paper focuses on the continuous time case and it allows approximation as opposed to exact matching of the input-output behavior. In addition,
the current paper uses the framework of  ``nice selections''. This framework is not only more general, but it has a clear system theoretical interpretation.

In the linear case, model reduction is a mature research area \cite{antoulas}. The subject of model reduction for hybrid and switched systems was addressed in several papers \cite{French1,Zhang20082944,Mazzi1,Chahlaoui,Habets1,China2,China3,Lam1,Kotsalis2,Kotsalis1,6209392,shaker2011}.
Except \cite{Habets1}, the cited papers propose techniques which involve solving certain LMIs, and for this reason, they tend to be applicable only to switched systems for which the continuous subsystems are stable.
In contrast, the approach of this paper works for systems which are unstable. However, this comes at a price, since we are not able to propose analytic error bounds, like the ones for balanced truncation \cite{petreczky2013}. In addition, the time horizon on which the approximation is ``good enough'', depends on the LSS. From a practical point of view, the lack of an analytic error bound and related issues need not be a very serious disadvantage, since it is often acceptable to evaluate the accuracy of the approximation after the reduced model has been computed.


The model reduction algorithm proposed in this paper is similar in spirit to moment matching for linear systems \cite{antoulas,gugercin} and bilinear systems \cite{BilinearMomentMatching2,BilinearMomentMatching3,BilinearMomentMatching5}; however, the details and the system class considered are entirely different. The concept of nice selection of columns (resp. rows) of the reachability (resp. observability) matrix for model reduction of multi input - multi output (MIMO) linear systems appeared in \cite{gugercin}. The method presented in this paper is based on the generalization of this concept to LSSs. In fact, this is seen as another contribution of the present paper. The model reduction algorithm for LPV systems described in \cite{toth2012} is related to the method given in this paper, as it also relies on a realization algorithm and Markov parameters. In turn, the realization algorithms and Markov parameters of LPV systems and LSSs are closely related, \cite{PM12}. However, the algorithm of \cite{toth2012} applies to a different system class (namely LPV systems), and it is not yet clear if it yields a partial realization of the original system considered.

\textbf{Outline}
In Section \ref{sect:prelim}, we fix the notation and terminology of the paper. In Section \ref{sect:lin_switch_def}, we present the formal definition and main properties of LSSs. In Section \ref{sect:markov}, we recall the concept of Markov parameters for linear systems and LSSs, and the problem of model reduction by moment matching. The solution to the moment matching problem for LSSs analogous to the linear case is stated in \ref{sect:Npart}. This solution is generalized and made useful further for LSSs in Section \ref{sect:nice} where also the related algorithm is stated in detail. Finally, in Section \ref{sec:exam} the two methods are illustrated on numerical examples.

%
%

%
%

\section{Preliminaries: notation and terminology}
\label{sect:prelim}

Denote by $\mathbb{N}$ the set of natural numbers including $0$. Denote by $\mathbb{R}_+$ the set $[0,+\infty)$ of nonnegative real numbers. In the sequel, let $PC(\mathbb{R}_+,S)$, with $S$ a topological subspace of an Euclidean space $\mathbb{R}^{n}$, denote the set of \emph{piecewise-continuous and left-continous maps}. That is, $f\in PC(\mathbb{R}_+,S)$ if it has finitely many points of discontinuity on any compact subinterval of $\mathbb{R}_+$, and at any point of discontinuity both the left-hand and right-hand side limits exist, and $f$ is continuous from the left. Moreover, when $S$ is a discrete set it will always be endowed with the discrete topology.

In addition, denote by $AC(\mathbb{R}_+, \mathbb R^n)$ the set of \emph{absolutely continuous maps}, and $L_{loc}(\mathbb{R}_+, \mathbb R^n)$ the set of \emph{Lebesgue measurable maps} which are integrable on any compact interval.

If $M \in \mathbb{R}^{a \times b}$ with $a,b \in \mathbb{N} \backslash \{0\}$ is a real matrix (or vector), $M_{i,:}$ (resp. $M_{:,j}$) denotes the $i$th row of $M$ with $i \in \{1, \dots, a\}$ (resp. $j$th column of $M$ with $j \in \{1, \dots, b\}$). The notation $M_{i,j}$ is used for addressing the entry of $M$ in its $i$th row and $j$th column. Lastly, $e_i$ will be used to denote the $i$th unit vector in the canonical basis for $\mathbb{R}^a$.

\section{Linear switched systems}
\label{sect:lin_switch_def}

In this section, we present the formal definition of linear switched systems and recall a number of relevant definitions. We follow the presentation of \cite{MP:BigArticlePartI,petreczky2013}.

\begin{Definition}[LSS] \label{LSS}
	A continuous time linear switched system (LSS) is a control system of the form
	\begin{subequations}\label{eq:LSSform}
		\begin{align} 
			\frac{d}{dt}x(t)&=A_{\sigma(t)}x(t)+B_{\sigma(t)}u(t),\quad x(t_0)=x_0\label{sdyn} \\
			y(t)&=C_{\sigma(t)}x(t)\label{out} 
		\end{align}
	\end{subequations}
	where $ \sigma \in PC(\mathbb{R}_+, Q)$ is the switching signal, $u \in L_{loc}(\mathbb{R}_+, \mathbb R^m)$ is the input, $x \in AC(\mathbb{R}_+, \mathbb R^n)$ is the state, and $y \in PC(\mathbb{R}_+, \mathbb R^p)$ is the output and $Q=\{1,\dots,D\},~D>0,$ is the set of discrete modes. Moreover, $A_q \in \mathbb{R}^{n \times n}$, $B_q \in \mathbb{R}^{n \times m}$, $C_q \in \mathbb{R}^{p \times n}$ are the matrices of the linear system in mode $q \in Q$, and $x_0$ is the initial state. The notation
	\begin{equation}
		\Sigma=(p,m,n,Q,\{(A_q,B_q,C_q)|q \in Q\},x_0) 
	\end{equation}
	or simply $\Sigma$, are used as short-hand representations for an LSSs of the form \eqref{eq:LSSform}. The number $n$ is the \emph{dimension (order) of $\Sigma$} and will sometimes be denoted by $\dim(\Sigma)$.
\end{Definition}

Next, we present the basic system theoretic concepts for LSSs.

\begin{Definition}\label{def:stateandoutput}
	The \emph{input-to-state} map $X_{\Sigma,x}$ and \emph{input-to-output} map $Y_{\Sigma,x}$ of $\Sigma$ are the maps 
	\begin{align*}
		X_{\Sigma,x}: L_{loc}(\mathbb{R}_+,\mathbb{R}^m) \times PC(\mathbb{R}_+,Q) & \rightarrow AC(\mathbb{R}_+,\mathbb{R}^n); \\
		(u,\sigma) & \mapsto X_{\Sigma,x}(u,\sigma), \\ 
		Y_{\Sigma,x}: L_{loc}(\mathbb{R}_+,\mathbb{R}^m) \times PC(\mathbb{R}_+,Q) & \rightarrow PC(\mathbb{R}_+,\mathbb{R}^p); \\
		(u,\sigma) & \mapsto Y_{\Sigma,x}(u,\sigma).
	\end{align*}
	defined by letting $t\mapsto X_{\Sigma,x}(u,\sigma)(t)$ be the solution to the Cauchy problem \eqref{sdyn} with $t_0=0$ and $x_0=x$, and letting $Y_{\Sigma,x}(u,\sigma)(t)=C_{\sigma(t)}X_{\Sigma,x}(u,\sigma)(t)$ as in \eqref{out}.
\end{Definition}

The input-output behavior of an LSS realization can be formalized as a map
\begin{equation} \label{eq:inputoutputmap}
	f: L_{loc}(\mathbb{R}_+,\mathbb{R}^m) \times PC(\mathbb{R}_+,Q) \rightarrow PC(\mathbb{R}_+,\mathbb{R}^p).
\end{equation}
The value $f(u,\sigma)$ represents the output of the underlying (black-box) system. This system may or may not admit a description by an LSS. Next, we define when an LSS describes (realizes) a map of the form \eqref{eq:inputoutputmap}.

The LSS $\Sigma$ of the form \eqref{eq:LSSform} is a \emph{realization} of an input-output map $f$ of the form \eqref{eq:inputoutputmap}, if $f$ is the input-output map of $\Sigma$ which corresponds to the initial state $x_0$, i.e., $f=Y_{\Sigma,x_0}$. The map $Y_{\Sigma,x_0}$ will be referred to as the \emph{input-output map of} $\Sigma$.

Moreover, we say that the LSSs $\Sigma_1$ and $\Sigma_2$ are \emph{equivalent} if $Y_{\Sigma_1,x_0^1}=Y_{\Sigma_2,x_0^2}$ where $x_0^1$ and $x_0^2$ denote the initial states of $\Sigma_1$ and $\Sigma_2$ respectively. The LSS $\Sigma_{\mathrm m}$ is said to be a \emph{minimal} realization of $f$, if $\Sigma_{\mathrm m}$ is a realization of $f$, and for any LSS $\Sigma$ such that $\Sigma$ is a realization of $f$, $\dim (\Sigma_{\mathrm m}) \le \dim (\Sigma)$. An LSS $\Sigma$ is said to be \emph{observable}, if for any two states $x_{1}\neq x_{2} \in \mathbb{R}^n$, $Y_{\Sigma,x_1} \neq Y_{\Sigma,x_2}$.

Let $\mathrm{Reach}_{x_0}(\Sigma)\subseteq \mathbb{R}^n$ denote the reachable set of the LSS $\Sigma$ relative to the initial condition $x_0\in \mathbb{R}^n$, i.e., $\mathrm{Reach}_{x_0}(\Sigma)$ is the image of the map $(u,q,t)\mapsto X_{\Sigma,x_0}(u,q)(t)$. The LSS $\Sigma$ is said to be \emph{span reachable} if the linear span of states which are reachable from the initial state is $\mathbb{R}^n$, i.e., if $\SPAN\{x \mid x \in \mathrm{Reach}_{x_0}(\Sigma) \}=\mathbb{R}^n$. Span-reachability, observability and minimality are related as follows.
\begin{Theorem}[\cite{MP:BigArticlePartI}]
	An LSS $\Sigma$ is a minimal realization of $f$ if and only if it is a realization of $f$, and it is span-reachable and observable. If $\Sigma_1 = (p,m,n,Q,\{(A_q,B_q,C_q)|q \in Q\},x_0)$ and $\Sigma_{2}= (p,m,n,Q,\{(A_q^a,B_q^a,C_q^a)|q \in Q\},x_0^a)$ are two minimal realizations of $f$, then they are \emph{isomorphic}, i.e., there exists a non-singular $S \in \mathbb{R}^{n \times n}$ such that
	\begin{equation*}
		Sx_0=x_0^a \mbox{ and } \forall q \in Q:
		A^{a}_{q}S=SA_{q},  B_{q}^{a}=SB_{q},
		C_{q}^{a}S=C_{q}.
	\end{equation*}
\end{Theorem}
Moreover, if $\Sigma$ is a realization of $f$, then there exists an algorithm for computing from $\Sigma$ a minimal realization $\Sigma_m$ of $f$, \cite{MP:BigArticlePartI,petreczky2013}. Hence, in the sequel, unless stated otherwise we will tacitly assume that the LSSs are minimal realizations of their input-output maps.

\section{Background on Markov parameters and moment matching}
\label{sect:markov}

In this section, we recall the concepts of Markov parameters and moment matching for linear systems and draw the analogy with the linear switched case. We will begin by recalling model reduction by moment matching for linear systems \cite{antoulas}.

\subsection{Markov parameters and moment matching for linear systems}

Recall that a potential input-output map of a linear system is an affine map $f: L_{loc}(\mathbb{R}_+,\mathbb{R}^m) \rightarrow PC(\mathbb{R}_+,\mathbb{R}^p)$ for which there exist analytic functions $K:\mathbb{R}_{+} \rightarrow \mathbb{R}^{p}$ and $G:\mathbb{R}_{+} \rightarrow \mathbb{R}^{p \times m}$, such that
\begin{equation}
	\label{rev1} 
	f(u)(t)=K(t)+\int_{0}^{t} G(t-s)u(s)ds, \forall t \in \mathbb{R}_{+}
\end{equation}
for all $u \in  L_{loc}(\mathbb{R}_+,\mathbb{R}^m)$. Existence of such a pair of maps is a necessary condition for $f$ to be realizable by a linear system. Indeed, consider a linear system
\begin{equation}
	\label{rev2}
	\Sigma\left\{
	\begin{split}
		& \dot x(t) = Ax(t)+Bu(t) \mbox{, where } x(0)=x_0 \\
		& y(t)=Cx(t)
	\end{split}\right.
\end{equation}
where $A$, $B$ and $C$ are $n \times n$, $n \times m$ and $p \times n$ real matrices and $x_0 \in \mathbb{R}^n$ is the initial state. The map $f$ is said to be \emph{realized by} $\Sigma$, if the output response at time $t$ of $\Sigma$ to any input $u$ equals $f(u)(t)$. This is the case if and only if $f$ is of the form \eqref{rev1} with $K(t)=Ce^{At}x_0$ and $G(t)=Ce^{At}B$.

If $f$ is of the form \eqref{rev1}, then $f$ is uniquely determined by the analytic functions $K$ and $G$. In turn, these functions are uniquely determined by their Taylor-coefficients at zero. Consequently, it is reasonable to approximate $f$ by the function
\[
\bar{f}(u)(t)=\bar{K}(t)+\int^{t}_{0} \bar{G}(t-s)u(s)ds,
\]
such that the first $N+1$ Taylor series coefficients of $\bar{K},\bar{G}$ and $K,G$ coincide, i.e., $\frac{d^k}{dt} K(t)|_{t=0}=\frac{d^k}{dt} \bar{K}(t)|_{t=0}$ and $\frac{d^k}{dt} G(t)|_{t=0}=\frac{d^k}{dt} \bar{G}(t)|_{t=0}$ for all $k=0,\ldots, N$.
The larger $N$ is, the more accurate the approximation is expected to be. One option is to choose $N$ and $\bar{f}$ in such a way that $\bar{f}$ would be realizable by an LTI (linear time invariant) state-space representation. In this case, this LTI state-space representation is called an \emph{N-partial realization of $f$}. Specifically, define the $k$th Markov parameter of $f$ as follows
\begin{equation}
	\label{rev1:markov}
	M_{k}=\begin{bmatrix} \frac{d^k}{dt^k} K(t)|_{t=0} ,
		\frac{d^k}{dt^k} G(t)|_{t=0} \end{bmatrix}, k \in \mathbb{N}.
\end{equation}
Note that if $K=0$ and $H(s)$ is the Laplace transform of $G$, then the Markov parameters are the coefficients of the Laurent expansion of $H(s)$, i.e., $H(s)=\sum_{i=1}^{\infty} M_is^{-i}$ for all $s \in \mathbb{C}$, $s \ne 0$. For the general case, if the linear system \eqref{rev2} is a realization of $f$, then the Markov-parameters can be expressed as $M_{k}=CA^k\begin{bmatrix} x_0 & B \end{bmatrix}$, for all $k \in \mathbb{N}$. Moreover, the linear system \eqref{rev2} is an $N$-partial realization of $f$, if $M_{k}=CA^k\begin{bmatrix} x_0 & B \end{bmatrix}$, $k=0,\ldots,N$. It can also be shown that if $f$ has a realization by an LTI system of order $N$, then the linear system \eqref{rev2} is a realization of $f$ if and only if it is a $2N-1$ partial realization of $f$, i.e., in this case $f$ is uniquely characterized by \emph{finitely many} Markov parameters.

The main idea behind model reduction of LTI systems using moment matching is as follows. Consider an LTI system $\Sigma$ of the form \eqref{rev2} and fix $N > 0$. Let $f$ be the input-output map of $\Sigma$ from the initial state $x_0$. Find an LTI system $\bar{\Sigma}$ of order $r$ strictly less than $n$ such that $\bar{\Sigma}$ is an $N$-partial realization of $f$. A relation between $r$ and $N$ will be discussed later in the paper.

There are several equivalent ways to interpret the relationship between the LTI systems $\Sigma$ and $\bar{\Sigma}$. Assume that the system matrices of $\bar{\Sigma}$ are $\bar{A},\bar{B},\bar{C}$ and the initial state of $\bar{\Sigma}$ is $\bar{x}_0$. If $\bar{\Sigma}$ is a solution to the moment matching problem described above, then the first $N+1$ coefficients of the Laurent series expansion of the transfer functions $C(sI-A)^{-1}\begin{bmatrix} x_0 & B \end{bmatrix}$ and $\bar{C}(sI-\bar{A})^{-1}\begin{bmatrix} \bar{x}_0 & \bar{B} \end{bmatrix}$ coincide. Yet another way to interpret the LTI system $\bar{\Sigma}$ is to notice that $CA^k\begin{bmatrix} x_0 & B \end{bmatrix} = \bar{C}\bar{A}^k\begin{bmatrix} \bar{x}_0 & \bar{B} \end{bmatrix}$ for all $k=0,\ldots,N$.

\subsection{Markov parameters and moment matching for linear switched systems}

In this paper, we will extend the idea of moment matching from LTI systems to LSSs. To this end, we will use the generalization of Markov parameters to the input-output maps of LSSs.

\begin{Notation} \label{Notation1}
	Consider a finite non-empty set $Q$ with $D$ elements, which will be called the \emph{alphabet}. Denote by $Q^*$ the set of finite sequences of elements of $Q$. The elements of $Q^*$ are called \emph{strings} or \emph{words} over $Q$ and any set $L \subseteq Q^*$ is called a \emph{language} over $Q$. Each non-empty word $w$ is of the form $w=q_1q_2 \cdots q_k$ for some $q_1,q_2,\dots,q_k \in Q$. The element $q_i$ is called the \emph{$i$th letter of $w$}, for $i=1,2,\dots,k$, and $k$ is called the \emph{length} of $w$. The \emph{empty sequence (word)} is denoted by $\varepsilon$. The length of word $w$ is denoted by $|w|$; note that $|\varepsilon|=0$. The set of non-empty words is denoted by $Q^+$, i.e., $Q^+=Q^*\backslash\{\varepsilon\}$. The subset of $Q^*$ containing all the words of length at most (resp. at least) $N \in \mathbb{N}$ will be denoted by $Q^{\leq N}$ (resp. $Q^{\geq N}$). The \emph{concatenation} of word $w \in Q^*$ with $v \in Q^*$ is denoted by $wv$: If $v=v_{1}v_2 \cdots v_{k}$, and  $w=w_{1}w_2 \cdots w_{m}$, $k > 0, m >0$, then $vw=v_{1}v_2 \cdots v_{k}w_{1}w_2 \cdots w_{m}$. If $v=\epsilon$, then $wv=w$; if $w=\epsilon$, then $wv=v$. For simplicity, the finite set $Q$ will be identified with its index set, that is $Q=\{1,2,\dots,D\}$.
\end{Notation}

Next consider an input-output map $f$ of the form \eqref{eq:inputoutputmap}. Notice that the restriction to a finite interval $[0,t]$ of any $\sigma \in PC(\mathbb{R}_+,Q)$ can be interpreted as finite sequence of elements from $Q \times \mathbb{R}_+$ of the form
\begin{equation} \label{eq:switching_sequence}
	\mu=(q_1,t_1)(q_2,t_2)\cdots(q_k,t_k)
\end{equation}
where $q_1,\dots,q_k \in Q$ and $t_1,\dots,t_k \in \mathbb{R_+}\backslash \{0\}$, $t_1+\cdots +t_k = t$, such that for all $s \in[0,t]$
\begin{equation} \label{eq:switch_seq_sigma_def}
\sigma(s)=\left\{\begin{array}{rl}
q_1 & \mbox{ if } s \in [0,t_1] \\
q_2 & \mbox{ if } s \in (t_1,t_1+t_2] \\
\vdots \\
q_i & \mbox{ if } s \in (t_1+\cdots t_{i-1} ,t_1+\cdots + t_{i-1}+t_i] \\
\vdots \\
q_k & \mbox{ if } s \in (t_1+\cdots t_{k-1} ,t_1+\cdots + t_{k-1}+t_k] \\
\end{array}\right.
\end{equation}
Clearly this encoding is not one-to-one, since if $q_{i-1}=q_i$ for any $i \in \{ 2,\ldots,k \}$ and $\mu=(q_1,t_1)(q_2,t_2)\cdots(q_k,t_k)$ corresponds to $\sigma|_{[0,t]}$, then $(q_1,t_1)(q_2,t_2)\cdots (q_{i-1},t_{i-1}+t_i)(q_{i+1},t_{i+1}) \cdots (q_k,t_k)$ also corresponds to $\sigma|_{[0,t]}$.

From \cite{MP:BigArticlePartI}, it follows that a necessary condition for $f$ to be realizable by an LSS is that $f$ has a \emph{generalized kernel representation}. For a detailed definition of a generalized kernel representation of $f$, we refer the reader to \cite[Definition 19]{MP:BigArticlePartI}.
\footnote{Note that in \cite{MP:BigArticlePartI} the concept of generalized kernel representation was defined for families of input-output maps. In order to apply the definition and results of \cite{MP:BigArticlePartI} to the current paper, one has to take a family of input-output maps $\Phi$ which is the family consisting of one single map $f$, i.e., $\Phi=\{f\}$. In addition, in \cite{MP:BigArticlePartI} the input-output maps were defined not for switching signals from $PC(\mathbb{R}_+,Q)$, but for switching sequences of the form \eqref{eq:switching_sequence}, where the times $t_1,\ldots,t_k$ were allowed to be zero. However, by using the correspondence between switching signals from $PC(\mathbb{R}_+,Q)$ and switching sequences \eqref{eq:switching_sequence}, and by using the properties (2) and (3) of \cite[Definition 19]{MP:BigArticlePartI}, we can easily adapt the definition and results from \cite{MP:BigArticlePartI} to the setting of the current paper.}

For our purposes, it is sufficient to recall that if $f$ has a generalized kernel representation, then there exists a unique family of analytic functions $K^f_{q_1,\ldots,q_k}:\mathbb{R}_{+}^k \rightarrow \mathbb{R}^p$ and $G^f_{q_1,\ldots,q_k}:\mathbb{R}_{+}^k \rightarrow \mathbb{R}^{p \times m}$, $q_1,\ldots,q_k \in Q$, $k \ge 1$, such that for all $(u,\sigma) \in L_{loc}(\mathbb{R}_{+},\mathbb{R}^{m}) \times PC(\mathbb{R}_+,Q)$, $t > 0$ and for any $\mu=(q_1,t_1)(q_2,t_2)\cdots(q_k,t_k)$ which corresponds to $\sigma$,
\begin{equation} \label{eq:inputoutput}
	\begin{aligned}
		& f(u,\sigma)(t)= K^f_{q_1q_2 \cdots q_k}(t_1,t_2, \dots, t_k)+ \\ 
		& \sum_{i=1}^{k} \int_{0}^{t_i}G^f_{q_iq_{i+1} \cdots q_k}(t_i-s,t_{i+1}, \dots, t_k)u \left( s+ \sum_{j=1}^{i-1}t_j \right)ds,
	\end{aligned}
\end{equation}
and the functions $\{K^f_{q_1\cdots q_k}, G^f_{q_1\cdots q_k} \mid q_1,\ldots,q_k \in Q, k \ge 0\}$ satisfy a number of technical conditions, see \cite[Definition 19]{MP:BigArticlePartI} for details.

From \cite{MP:BigArticlePartI}, it follows that there is a one-to-one correspondence between $f$ and the family of maps $\{K^f_{q_1\cdots q_k}, G^f_{q_1\cdots q_k} \mid q_1,\ldots,q_k \in Q, k \ge 0\}$. The maps $\{K^f_{q_1\cdots q_k}, G^f_{q_1\cdots q_k} \mid q_1,\ldots,q_k \in Q, k \ge 0\}$ play a role which is similar to the role of the functions $K$ and $G$ in the LTI case. If $f$ has a realization by an LSS \eqref{LSS}, then the functions $K^f_{q_1q_2 \cdots q_k}(t_1,t_2, \dots, t_k)$ and $G^f_{q_1q_2 \cdots q_k}(t_1,t_2, \dots, t_k)$ satisfy
\begin{align*}
	& K^f_{q_1q_2 \cdots q_k}(t_1,t_2, \dots, t_k)= C_{q_k}e^{A_{q_k}t_k}e^{A_{q_{k-1}}t_{k-1}} \cdots e^{A_{q_1}t_1}x_0 \\
	& G^f_{q_1q_2 \cdots q_k}(t_1,t_2, \dots, t_k)= C_{q_k}e^{A_{q_k}t_k}e^{A_{q_{k-1}}t_{k-1}} \cdots e^{A_{q_1}t_1}B_{q_1}.
\end{align*}
We can now define the Markov parameters of $f$ as follows.

\begin{Definition}[Markov parameters] \label{MarkovParameters}
	The Markov parameters of $f$ are the values of the map
	\[ M^f:Q^{*} \rightarrow \mathbb{R}^{\QNUM p \times (m\QNUM + 1)}, \]
	defined by
	\[ M^f(v)=\begin{bmatrix}
	S_0(v1) & S(1v1) & \cdots & S(\QNUM v1) \\
	S_0(v2) & S(1v2) & \cdots & S(\QNUM v2) \\
	\vdots  & \vdots   & \cdots & \vdots \\
	S_0(v\QNUM) & S(1v\QNUM) & \cdots & S(\QNUM v \QNUM) \\
	\end{bmatrix}, \]
	where the vectors $S_0(vq) \in \mathbb{R}^p$ and the matrices $S(q_0vq) \in \mathbb{R}^{p \times m}$ are defined as follows. For all $q_0,q \in Q$,
	\begin{align*}
		& S_0(q)= K^f_{q}(0) \mbox{ and }  S(q_0q)= G^f_{q_0q}(0,0).
	\end{align*}
	and for all $q_0,q \in Q$, $v \in Q^*$, $v \neq \varepsilon$ by
	\begin{align*}
		& S_0(vq)= \left. \frac{d}{dt_1} \cdots \frac{d}{dt_k} K^f_{q_1 \cdots q_k q}(t_1, \dots, t_k, 0) \right| _{t_1=t_2= \cdots =t_k=0} \\
		& S(q_0vq)=\left. \frac{d}{dt_1} \cdots \frac{d}{dt_k} G^f_{q_0q_1 \cdots q_k q}(0,t_1, \dots, t_k, 0) \right| _{t_1=t_2= \cdots =t_k=0}
	\end{align*}
	where $v=q_1q_2 \cdots q_k$, $k \geq 0$, $q_1,q_2, \dots, q_k \in Q$.
\end{Definition}

That is, the Markov parameters of $f$ are certain partial derivatives of the functions $\{K^f_{q_1\cdots q_k}, G^f_{q_1\cdots q_k} \mid q_1,\ldots,q_k \in Q, k \ge 0\}$. From \cite{MP:BigArticlePartI}, it follows that the Markov parameters $\{M^f(v)\}_{v \in Q^*}$ determine the maps $\{K^f_{q_1\cdots q_k}, G^f_{q_1\cdots q_k} \mid q_1,\ldots,q_k \in Q, k \ge 0\}$, and hence $f$, uniquely.
If $f$ has a realization by an LSS $\Sigma$ of the form \eqref{LSS}, then the Markov-parameters of $f$ can be expressed as products of the matrices of $\Sigma$. In order to present the corresponding formula, we will use the following notation.

\begin{Notation} \label{Notation2}
	Let $w=q_1q_2 \cdots q_k \in Q^*$, $q_1,\dots,q_k \in Q$, $k>0$ and $A_{q_i} \in \mathbb{R}^{n \times n}$, $i=1,\dots,k$. Then the matrix $A_w$ is defined as
	\begin{equation} 
		A_w=A_{q_k}A_{q_{k-1}}\cdots A_{q_1}.
	\end{equation}
	If $w= \varepsilon$, then $A_\varepsilon$ is the identity matrix.
\end{Notation}


From \cite{MP:BigArticlePartI}, it follows that an LSS \eqref{LSS} is a realization of the map $f$ if and only if $f$ has a generalized kernel representation and $S_0(vq)=C_qA_vx_0 \mbox{ and } S(q_0vq)=C_qA_vB_{q_0}$ for all $v \in Q^{*}$, or in more compact form
\begin{equation}
	\label{eq:markov1}
	M^f(v)=\widetilde{C}A_v\widetilde{B}, \mbox{ } \forall v \in Q^{*}
\end{equation}
with $\widetilde{C}=\begin{bmatrix} C_1^{\mathrm{T}} & \cdots & C_{\QNUM}^{\mathrm{T}} \end{bmatrix}^{\mathrm{T}}$ and $\widetilde{B}= \begin{bmatrix} x_0 & B_1 & B_2 & \cdots & B_D \end{bmatrix}$.
The main idea behind moment matching for LSSs (more precisely, for their input-output maps), is as follows: approximate $f$ by another input-output map $\bar{f}$, such that some of the Markov parameters of $f$ and $\bar{f}$ coincide. One obvious choice is to say that $M^{f}(v)=M^{\bar{f}}(v)$ for all $v \in Q^{*}$, $|v| \le N$ for some $N$. This approach will be explained in detail in the next section after formally defining $N$-partial realizations for an LSS. The other approach is based on the concept of nice selections of the columns (resp. rows) of the partial reachability (resp. observability) matrix of an LSS, and it will be presented in Section \ref{sect:nice}. The approach based on nice selections is less conservative and, as seen in Section \ref{sect:nice}, it can be used for matching the input output behavior of two LSSs along a certain switching sequence.

\section{Model reduction by $N$ or $2N$-partial realizations} \label{sect:Npart}

In this section, the aim is to present an efficient  model reduction algorithm which transforms an LSS $\Sigma$ into an LSS $\bar{\Sigma}$ such that $\dim (\bar{\Sigma}) \le \dim (\Sigma)$ and some number of Markov parameters of $\Sigma$ and $\bar{\Sigma}$ are equal. Firstly, we will formally define the concept of $N$-partial realizations and state the problem taken at hand in this section.


\begin{Definition}[$N$-partial realization] \label{def:Npartial}
	The LSS \eqref{LSS} is called \emph{$N$-partial realization} of $f$, if
	\[ M^f(v)=\widetilde{C}A_v\widetilde{B} \mbox{ } \forall v \in Q^{*}, |v| \le N: \]
	with $\widetilde{C}=\begin{bmatrix} C_1^{\mathrm{T}} & \cdots & C_{\QNUM}^{\mathrm{T}} \end{bmatrix}^{\mathrm{T}}$ and $\widetilde{B}=\begin{bmatrix} x_0 & B_1 & B_2 & \cdots & B_{\QNUM} \end{bmatrix}$.
\end{Definition}

If $\Sigma$ is of the form \eqref{LSS} and $Y_{\Sigma,x_0}$ is the input-output map of $\Sigma$, then the concept of $N$-partial realization can be interpreted as follows: $\Sigma$ is an $N$-partial realization of $f$, if those Markov parameters of $f$ and $Y_{\Sigma,x_0}$ which are indexed by words of length at most $N$ coincide. The analogous (to the linear case) problem of model reduction by moment matching for LSSs can now be formulated as follows.

\begin{Problem} \label{prob:momentmatching}
	\emph{($N$-Moment matching problem for an LSS).}
	Let $\Sigma$ be an LSS of the form \eqref{LSS} and let $f=Y_{\Sigma,x_0}$ be its input-output map. Fix $N \in \mathbb{N}$. Find an LSS $\bar{\Sigma}$ such that $\dim (\bar{\Sigma}) < \dim (\Sigma)$ and $\bar{\Sigma}$ is an $N$-partial realization of $f=Y_{\Sigma,x_0}$.
\end{Problem}

An $N$-partial realization $\bar{\Sigma}$ of $f$ means that all the partial derivatives of order at most $N$ of $\{ K^{f}_{q_1\cdots q_k}, G^{f}_{q_1\cdots q_k} \mid q_1,\ldots,q_k \in Q, k \ge 0 \}$ and of $\{ K^{\bar{f}}_{q_1\cdots q_k}, G^{\bar{f}}_{q_1\cdots q_k} \mid q_1,\ldots,q_k \in Q, k \ge 0 \}$ coincide, where $\bar{f}=Y_{\bar{\Sigma},\bar{x}_0}$. Intuitively, this will mean that for any input and switching signal $(u,\sigma) \in L_{loc}(\mathbb{R}_+,\mathbb{R}^m) \times PC(\mathbb{R}_+,Q)$, the outputs $f(u,\sigma)(t)$ and $\bar{f}(u,\sigma)(t)$ are close, for small enough $t$. In fact, this approach is the direct analogue of the moment matching methods for linear systems and it has a system theoretical interpretation. Namely, the following corollary of \cite[Theorem 4]{petreczky} clarifies this interpretation by stating how many Markov parameters of a map $f$ must be matched by an LSS $\bar{\Sigma}$, for it to be a realization of $f$. Note that there is a trade off between the choice of $N$ and the dimension $\Sigma$.

\begin{Corollary} \label{cor:rel_N_n}
	Assume that $\Sigma$ is a minimal realization of $f$ and $N$ is such that $2\dim (\Sigma)-1 \le N$. Then for any LSS $\bar{\Sigma}$ which is an $N$-partial realization of $f$, $\bar{\Sigma}$ is also a realization of $f$ and $\dim (\Sigma) \le \dim (\bar{\Sigma})$.
\end{Corollary}

That is, if we choose $N$ too high, namely if we choose any $N$ such that $N \geq 2n-1$, where $n$ is the dimension of a minimal LSSs realization of $f$, then there will be no hope of finding an LSS which is an $N$-partial realization of the original input-output map, and whose dimension is lower than $n$.

In order to solve Problem \ref{prob:momentmatching}, one could consider applying the partial realization algorithm \cite{petreczky}. In a nutshell, \cite{petreczky} defines finite \emph{Hankel matrices} and proposes a Kalman-Ho like realization algorithm based on the factorization of the Hankel matrix, \cite[Algorithm 1]{petreczky}. The problem with this naive approach is that it involves explicit construction of Hankel matrices, whose size is exponential in $N$. Consequently, the application of the partial realization algorithm would yield a model reduction algorithm whose memory-usage and run-time complexity is exponential. In the next section, we present a model reduction algorithm which yield a partial realization of the input-output map of the original system, and which does not involve the explicit computation of the Hankel matrix.



In the sequel, the image (column space) of a real matrix $M$ is denoted by $\IM (M)$ and $\Rank (M)$ is the dimension of $\IM (M)$. 


We will start with presenting the following definitions.

\begin{Definition} \label{ObservabilityMat}
	\emph{(Partial unobservability space).}
	The partial unobservability space $\mathscr{O}_N$ of $\Sigma$ up to words of length $N$ is defined as follows:
	\begin{equation}
		\mathscr{O}_N(\Sigma)= \bigcap_{v \in Q^{\leq N}} \ker(\widetilde{C}A_v).
	\end{equation}
\end{Definition}

In the rest, we will denote $\mathscr{O}_N(\Sigma)$ by $\mathscr{O}_N$ if $\Sigma$ is clear from the context. It is not difficult to see that $\mathscr{O}_0=\bigcap_{q \in Q} \ker (C_q)$ and for any $N > 0$, $\mathscr{O}_{N}=\mathscr{O}_{0} \cap \bigcap_{q \in Q} \ker (\mathscr{O}_{N-1}A_q)$. From \cite{Sun:Book,MP:BigArticlePartI}, it follows that $\Sigma$ is observable if and only if $\mathscr{O}_{N}=\{0\}$ for all $N \ge n-1$.

\begin{Definition}[Partial reachability space] \label{ReachabilityMat}
	The partial reachability space $\mathscr{R}_N$ of $\Sigma$ up to words of length $N$ is defined as follows:
	\begin{equation}
		\mathscr{R}_N(\Sigma)= \SPAN \{ \IM(A_v \widetilde{B}) \mid v \in Q^{\leq N} \}.
	\end{equation}
\end{Definition}

In the rest, we will denote $\mathscr{R}_N(\Sigma)$ by $\mathscr{R}_N$ if $\Sigma$ is clear from the context. It is easy to see that $\mathscr{R}_0=\IM (\widetilde{B})$ and $\mathscr{R}_{N}=\IM (\widetilde{B})+\sum_{q \in Q} \IM (A_q\mathscr{R}_{N-1})$, for $N > 0$ (note that here the summation operator must be interpreted as the Minkowski sum). It follows from \cite{MP:BigArticlePartI,Sun:Book} that $\Sigma$ is span-reachable if and only if $\dim (\mathscr{R}_N)=n$ for all $N \ge n-1$.

Given the definition of partial observability / reachability spaces, one can define the corresponding matrix representations $O_N$ and $R_N$ such that $\ker (O_N)= \mathscr{O_N}$ and $\IM(R_N)=\mathscr{R}_N$, and hence the partial Hankel matrix $H_{N,N}$ of an LSS $\Sigma$ as $H_{N,N}= O_N R_N$. Howevever, this is only a side remark since the methods given in this paper will not use explicit representations of the Hankel matrices.

\begin{Theorem} \label{theo:mert1}
	\emph{(One sided moment matching for $N$-partial realizations (reachability)).}
	Let
	\[
	\Sigma=(p,m,n,Q,\{(A_q,B_q,C_q)|q \in Q\},x_0)
	\]
	be an LSS realization of the input-output map $f$, $V \in \mathbb{R}^{n \times r}$ be a full column rank matrix such that
	\[
	\mathscr{R}_{N}(\Sigma) = \IM (V).
	\]
	If $\bar{\Sigma}=(p,m,r,Q,\{(\bar{A}_q,\bar{B}_q,\bar{C}_q)|q \in Q\},\bar{x}_0)$ is an LSS such that for each $q \in Q$, the matrices $\bar{A}_q,\bar{B}_q,\bar{C}_q$ and the vector $\bar{x}_0$ are defined as
	\[
	\bar{A}_q=V^{-1}A_qV \mbox{, } \bar{B}_q=V^{-1}B_q \mbox{, } \bar{C}_q=C_qV \mbox{, } \bar{x}_0=V^{-1}x_0,
	\]
	where $V^{-1}$ is a left inverse of $V$, 
	then $\bar{\Sigma}$ is an $N$-partial realization of $f$.
\end{Theorem}

\begin{proof}
	Let $w=q_1\cdots q_k$, $k=0,\ldots,N$, $q_1,\ldots,q_k \in Q$ and let $q_0 \in Q$. If $k=0$, then $w=\epsilon$. Since the conditions of Theorem~\ref{theo:mert1} imply $\IM (B_{q_0}) \subseteq \IM (V)$ and $V^{-1}$ is a left inverse of $V$, it is a routine exercise to see that $VV^{-1}B_{q_0}=B_{q_0}$.If $k > 0$, then $\IM (A_{q_i} \cdots A_{q_1}B_{q_0})$ is also a subset of $\mathscr{R}_N=\IM (V)$, $i=1,\ldots,k$. Hence, by induction we can show that 
	$VV^{-1}A_{q_i} \cdots A_{q_1}B_{q_0}=A_{q_i} \cdots A_{q_1}B_{q_0}$, $i=1,\ldots,k$, which ultimately yields
	\begin{equation} \label{eq2:theo_mert1}
		V\bar{A}_{q_{k}}\cdots \bar{A}_{q_1}\bar{B}_{q_0}=V\bar{A}_w\bar{B}_{q_0}=A_wB_{q_0}.
	\end{equation}
	Using a similar argument, we can show that
	\begin{equation} \label{eq21:theo_mert1}
		V\bar{A}_w\bar{x}_0=A_wx_0.
	\end{equation}
	Using \eqref{eq2:theo_mert1} and \eqref{eq21:theo_mert1}, and 
	$\bar{C}_q=C_qV$, $q \in Q$, we conclude that for all $w \in Q^{*}$, $|w| \le N$,
	$q,q_0 \in Q$, 
	\[ 
	\bar{C}_q\bar{A}_w\bar{B}_{q_0}=C_qA_wB_{q_0}
	\mbox{ and }
	\bar{C}_q\bar{A}_w\bar{x}_0 = C_qA_wx_0,
	\]
	from which the statement of the theorem follows.
\end{proof}

Note that the number $r$ is the number of columns in the full column rank matrix $V$, hence $r \leq n$. This fact leads $\bar{\Sigma}$ to be of reduced order if $N$ is sufficiently small, see Corollary \ref{cor:rel_N_n}. Using a dual argument, we can prove the following dual result.

\begin{Theorem}[One sided moment matching (observability)] \label{theo:mert2}
	Let $\Sigma=(p,m,n,Q,\{(A_q,B_q,C_q)|q \in Q\},x_0)$ be an LSS realization of the input-output map $f$, $W \in \mathbb{R}^{r \times n}$ be a full row rank matrix such that
	\[
	\mathscr{O}_{N}(\Sigma) = \ker (W)
	\]
	Let $W^{-1}$ be any right inverse of $W$ and let
	\[
	\bar{\Sigma}=(p,m,r,Q,\{(\bar{A}_q,\bar{B}_q,\bar{C}_q)|q \in Q\},\bar{x}_0)
	\]
	be an LSS such that for each $q \in Q$, the matrices $\bar{A}_q,\bar{B}_q,\bar{C}_q$ and the vector $\bar{x}_0$ are defined as
	\[
	\bar{A}_q=WA_qW^{-1} \mbox{, } \bar{B}_q=WB_q \mbox{, } \bar{C}_q=C_qW^{-1} \mbox{, } \bar{x}_0=Wx_0.
	\]
	Then $\bar{\Sigma}$ is an $N$-partial realization of $f$.
\end{Theorem}

Finally, by combining the proofs of Theorem \ref{theo:mert1} and Theorem \ref{theo:mert2}, we can show the following.

\begin{Theorem}[Two sided moment matching] \label{theo:mert3}
	Let $\Sigma=(p,m,n,Q,\{(A_q,B_q,C_q)|q \in Q\},x_0)$ be an LSS realization of the input-output map $f$, $V \in \mathbb{R}^{n \times r}$ and $W \in \mathbb{R}^{r \times n}$ be respectively full column rank and full row rank matrices such that
	\[
	\mathscr{R}_{N}(\Sigma) = \IM (V) \mbox{, } \mathscr{O}_{N}(\Sigma) = \ker (W) \mbox{ and } \Rank(WV)=r.
	\]
	If $\bar{\Sigma}=(p,m,r,Q,\{(\bar{A}_q,\bar{B}_q,\bar{C}_q)|q \in Q\},\bar{x}_0)$ is an LSS such that for each $q \in Q$, the matrices $\bar{A}_q,\bar{B}_q,\bar{C}_q$ and the vector $\bar{x}_0$ are defined as
	\[
	\bar{A}_q=WA_qV(WV)^{-1} \mbox{, } \bar{B}_q=WB_q \mbox{, } \bar{C}_q=C_qV(WV)^{-1} \mbox{, } \bar{x}_0=Wx_0,
	\]
	then $\bar{\Sigma}$ is a $2N$-partial realization of $f$.
\end{Theorem}

Note that having a $2N$-partial realization as an approximation system would be more desirable than having an $N$-partial realization, since number of matched Markov parameters would increase. However, it is only possible to get a $2N$-partial realization for the original system $\Sigma$ when the additional condition $\Rank(V)=\Rank(W)=\Rank(WV)=r$ is satisfied. Now, we will present an efficient algorithm of model reduction by moment matching, which computes either an $N$ or $2N$-partial realization $\bar{\Sigma}$ for an $f$ which is realized by an LSS $\Sigma$. First, we present algorithms for computing the subspaces $\mathscr{R}_N$ and $\mathscr{O}_N$. To this end, we will use the following notation: if $M$ is any real matrix, then denote by $\mathbf{orth}(M)$ the matrix $U$ such that $U$ is full column rank, $\IM (U)=\IM (M)$ and $U^{\mathrm{T}}U=I$. Note that $U$ can easily be computed from $M$ numerically, see for example the Matlab command \texttt{orth}.

The algorithm for computing $V  \in \mathbb{R}^{n \times r}$ such that $\IM(V)=\mathscr{R}_N$ is presented in Algorithm \ref{alg1} below.

\begin{algorithm}[h]
	\caption{
		Calculate  a matrix representation of $\mathscr{R}_N$,
		\newline
		\textbf{Inputs}: $(\{A_q,B_q\}_{q \in Q},x_0)$ and $N$
		\newline
		\textbf{Outputs:} $V  \in \mathbb{R}^{n \times r}$ such that $\Rank (V)=r$,
		$\IM (V) = \mathscr{R}_N$.
	}
	\label{alg1}
	\begin{algorithmic}
		\STATE $V:=U_0$, $U_0:=\mathbf{orth}\begin{bmatrix} x_0,& B_1, & \ldots, & B_{\QNUM} \end{bmatrix}$.
		\FOR{$k=1\ldots N$} 
		\STATE
		$V:=\mathbf{orth}(\begin{bmatrix} V, & A_1V, & A_2V, & \ldots, & A_{\QNUM}V \end{bmatrix})$
		\ENDFOR
		\RETURN $V$.
	\end{algorithmic}
\end{algorithm}

By duality, we can use Algorithm \ref{alg1} to compute a $W \in \mathbb{R}^{r \times n}$ such that $\ker(W)=\mathscr{O}_N$, the details are presented in Algorithm \ref{alg2}.

\begin{algorithm}
	\caption{
		Calculate a matrix representation of $\mathscr{O}_N$
		\newline
		\textbf{Inputs}: $\{A_q,C_q\}_{q \in Q}$ and $N$
		\newline
		\textbf{Output:} $W \in \mathbb{R}^{r \times n}$, such that
		$\Rank (W) = r$ and $\ker (W)=\mathscr{O}_N$.
	}
	\label{alg2}
	\begin{algorithmic}
		\STATE Apply Algorithm \ref{alg1} with inputs $(\{A_q^{\mathrm{T}},C_q^{\mathrm{T}}\}_{q \in Q},0)$ to obtain
		a matrix $V$.
		\RETURN $W=V^{\mathrm{T}}$.
	\end{algorithmic}
\end{algorithm} 

Notice that the computational complexity of Algorithm \ref{alg1} and Algorithm \ref{alg2} is polynomial in $N$ and $n$, even though the spaces of $\mathscr{R}_N$ (resp. $ \mathscr{O}_N$) are generated by images (resp. kernels) of exponentially many matrices.

Using Algorithms \ref{alg1} and \ref{alg2}, we can formulate a model reduction algorithm, see Algorithm \ref{alg3}.
\begin{algorithm}
	\caption{Moment matching for LSSs
		\newpage 
		\textbf{Inputs:} $\Sigma=(p,m,n,Q,\{(A_q,B_q,C_q)|q \in Q\},x_0)$, $\mathrm{Mode} \in \{ R,O,T \}$ and $N \in \mathbb{N}$.
		\newpage  
		\textbf{Output: } $\bar{\Sigma}=(p,m,r,Q,\{(\bar{A}_q,\bar{B}_q,\bar{C}_q)|q \in Q\},\bar{x}_0)$.
	}
	\label{alg3}
	\begin{algorithmic}
		\STATE Using Algorithm \ref{alg1}-\ref{alg2} compute matrices $V$ and $W$ such that
		$V$ is full column rank, $W$ is full row rank and $\IM (V)= \mathscr{R}_N$,
		$\ker (W) = \mathscr{O}_N$.
		\IF{$\Rank (V)=\Rank (W)=\Rank (WV)$ and $\mathrm{Mode}=T$} 
		\STATE
		Let $r=\Rank (V)$ and
		\begin{align*}
			& \bar{A}_q=WA_qV(WP)^{-1} \mbox{, } \bar{C}_q=C_qV(WV)^{-1} \mbox{, } \\
			& \bar{B}_q=WB_q \mbox{, } \bar{x}_0=Wx_0. 
		\end{align*}
		\ENDIF
		\IF{$\mathrm{Mode}=R$}
		\STATE
		Let $r=\Rank (V)$, $V^{-1}$ be a left inverse of $V$ and set
		\[
		\bar{A}_q=V^{-1}A_qV \mbox{, } \bar{C}_q=C_qV \mbox{, } \bar{B}_q=V^{-1}B_q \mbox{, } \bar{x}_0=V^{-1}x_0.
		\]
		\ENDIF
		\IF{$\mathrm{Mode}=O$}
		\STATE
		Let $r=\Rank (W)$ and let $W^{-1}$ be a right inverse of $W$. Set
		\[
		\bar{A}_q=WA_qW^{-1} \mbox{, } \bar{C}_q=C_qW^{-1} \mbox{, } \bar{B}_q=WB_q \mbox{, } \bar{x}_0=Wx_0.
		\]
		\ENDIF
		\RETURN $\bar{\Sigma}=(p,m,r,Q,\{(\bar{A}_q,\bar{B}_q,\bar{C}_q)|q \in Q\},\bar{x}_0)$.
	\end{algorithmic}
\end{algorithm}

Theorems \ref{theo:mert1} -- \ref{theo:mert3} imply the following corollary on correctness of Algorithm \ref{alg3}.

\begin{Corollary}[Correctness of Algorithm \ref{alg3}]
	Using the notation of Algorithm \ref{alg3}, the following holds: If $\Rank (V)=\Rank (W) = \Rank (WV)$ and $\mathrm{Mode}=T$, then Algorithm \ref{alg3} returns a $2N$-partial realization of $f=Y_{\Sigma,x_0}$ (if $\mathrm{Mode}=T$ and the rank condition does not hold, the algorithm returns nothing). Otherwise, Algorithm \ref{alg3} returns an $N$-partial realization of $f=Y_{\Sigma,x_0}$.
\end{Corollary}

Note that the input variable $\mathrm{Mode}$ in Algorithm \ref{alg3} represents the choice of the user on which method to be used, i.e., if $\mathrm{Mode}=R$, the algorithm uses Theorem \ref{theo:mert1}; if $\mathrm{Mode}=O$, the algorithm uses Theorem \ref{theo:mert2} and if $\mathrm{Mode}=T$, the algorithm uses Theorem \ref{theo:mert3}. If $\mathrm{Mode}=T$ and the condition $\Rank (V)=\Rank (W)=\Rank (WV)$ does not hold, Algorithm \ref{alg3} can always be used for getting an $N$-partial realization, by choosing $\mathrm{Mode}=O$ or  $\mathrm{Mode}=R$.

\section{Model reduction by nice selections}
\label{sect:nice}

In this section, a more general approach for moment matching of LSSs will be taken. In contrast to the $N$-partial realization solution, the material in this section is not direct analogue of the moment matching for linear systems, it is more suited for LSSs specifically. The notion of nice selections of columns (resp. rows) of the reachability (resp. observability) matrices of an LSS, gives flexibility to the user of the method in this section in the following sense. The user may focus on the approximation of some specific modes more than the others. Moreover, as we will show in Theorem \ref{thm:nice_sel_switch_seq}, the method can be used for exactly matching (or approximating) the input-output behavior of a continuous time LSS with an LSS of possibly lower order for a \emph{certain} switching sequence.

Now, the concept of nice column (resp. row) selections for (partial) reachability (resp. observability) space of an LSS will be defined. This is the central tool for the moment matching method to be presented. 

\begin{Definition}[Nice selections]
	A subset $\alpha$ of $Q^* \times Q \times I$, $I=\{1, \dots, p\}$ is called a nice row selection for an LSS $\Sigma$, if $\alpha$ has the following property; if $(\sigma v, q, i) \in \alpha$ for some $\sigma \in Q$, $v \in Q^*$, then $(v,q,i) \in \alpha$.
	
	Likewise, a subset $\beta$ of $(Q^* \times Q \times J) \cup Q^*$, $J=\{1, \dots, m\}$, is called a nice column selection for an LSS  $\Sigma$, if $\beta$ has the following property; if $(w \sigma, q, j) \in \beta$ for some $\sigma \in Q$, $w \in Q^*$, then $(w,q,j) \in \beta$; and if $w \sigma \in \beta$ for some $\sigma \in Q$, $w \in Q^*$, then $w \in \beta$.
\end{Definition}

The spaces related to a row nice selection $\alpha$ or a column nice selection $\beta$ can now be defined.
 
\begin{Definition}[$\alpha$-unobservability and $\beta$-reachability spaces]
Let  $\Sigma$ be a minimal realization of $Y_{\Sigma,x_0}$. Let $\alpha$ be a nice row selection and $\beta$ be a nice column selection related to $\Sigma$. Then the subspaces
\begin{equation*}
\begin{aligned}
& \mathscr{O}_\alpha(\Sigma)= \bigcap_{(v,q,i) \in \alpha} \ker (e_i^\mathrm{T}C_qA_v) \\
& \mathscr{R}_\beta(\Sigma) = \SPAN \{ \{ A_wB_qe_j \mid (w,q,j) \in \beta \} \cup \{ A_wx_0 \mid w \in \beta \} \}
\end{aligned}
\end{equation*}
will be called $\alpha$-unobservability and $\beta$-reachability spaces of $\Sigma$ respectively.
\end{Definition}
Similarly to the previous section, $\mathscr{O}_\alpha(\Sigma)$ and $\mathscr{R}_\beta(\Sigma)$ will be denoted by $\mathscr{O}_\alpha$ and $\mathscr{R}_\beta$ if $\Sigma$ is clear from the context. 

\begin{Example} \label{ex:nice_sel1}
	In order to illustrate the notion of a nice selection, let us consider the linear SISO case. Then $p=m=D=1$, and hence $J=I$ and $v \in  Q^{*}$ can be identified with its length, since $v=\overbrace{11\cdots 1}^{|v|-times}$. It then follows that an element $(v,q,i)$ of a nice selection is of the form $i=q=1$ and $v=\overbrace{11\cdots 1}^{|v|-times}$ and hence it can be identified with the natural number $|v| \in \mathbb{N}$. Then a nice selection $\alpha$ can be identified with a subset $\widetilde{\alpha} \subseteq \mathbb{N}$ with the property that if $0 < k \in \widetilde{\alpha}$, then $k-1 \in \widetilde{\alpha}$. For the MIMO linear case, $D=1$, and any sequence $v \in Q^{*}$ can be identified with its length as explained above. Then a nice column selection $\beta$ is a subset of $(\mathbb{N} \times \{1,\ldots,m\}) \cup \mathbb{N}$, such that if $(k,j) \in \beta, k > 0$, then $(k-1,j) \in \beta$ and if $k \in \beta$ then $k-1 \in \beta$. A similar characterization holds for nice row selections. That is, for the linear case, our definition of nice selections yields the classical concept \cite{Hazewinkel1}.
\end{Example}

The moment matching method for LSSs to be presented is based on constructing matrix representations of the $\beta$-reachability or $\alpha$-unobservability spaces of an LSS $\Sigma$, i.e., again constructing the matrices $V$ or $W$ such that $\IM (V)= \mathscr{R}_{\beta}$ and $\ker (W)= \mathscr{O}_{\alpha}$. For this purpose, it is crucial to find a basis for those spaces. The following lemma connects the notion of nice selections to this goal.

\begin{Theorem} \label{lem:nice_selections}
	Let $\Sigma$ be an LSS of the form 
	\[\Sigma=(p,m,n,Q,\{(A_q,B_q,C_q)|q \in Q\},x_0).\]
	For any $r<n$, there exists a nice column selection $\beta \subseteq (Q^* \times Q \times J) \cup Q^*$ (resp. row selection $\alpha \subseteq Q^* \times Q \times I$) such that $\dim (\mathscr{R}_\beta)=r$ (resp. $ \dim (\mathscr{O}_\alpha)=n-r$). 
\end{Theorem}

\begin{proof}
	See Appendix.
\end{proof}

Let $\widetilde{\mathscr{O}}_N$ denote the space $\widetilde{\mathscr{O}}_N= \SPAN \{ \IM ((\widetilde{C}A_v)^\mathrm{T}) \mid v \in Q^{\leq N} \}$ and $\widetilde{\mathscr{O}}_\alpha$ denote the space $\widetilde{\mathscr{O}}_\alpha= \SPAN \{ (e_i^\mathrm{T}C_qA_v)^\mathrm{T} \mid (v,q,i) \in \alpha \}$ for a nice row selection $\alpha$. Note that $\widetilde{\mathscr{O}}_\alpha$ is isomorphic to the orthogonal complement of $\mathscr{O}_{\alpha}$ and $\widetilde{\mathscr{O}}_N$ is isomorphic to the orthogonal complement of $\mathscr{O}_{N}$. From the proof of Theorem~\ref{lem:nice_selections}, it can be seen that there exists a nice column selection $\beta \subseteq (Q^{\le N} \times Q \times J) \cup Q^*$, and a nice row selection of $\alpha \subseteq Q^{\le N} \times Q \times I$, such that $\dim (\widetilde{\mathscr{O}}_{\alpha})=\dim (\widetilde{\mathscr{O}}_N)$, $\dim (\mathscr{R}_{\beta})=\dim (\mathscr{R}_N)$, and the vectors of $\widetilde{\mathscr{O}}_N$ indexed by the elements of $\alpha$ (respectively the vectors of $\mathscr{R}_N$ indexed by the elements of $\beta$) are linearly independent. It means that if $r_1=\dim (\widetilde{\mathscr{O}}_N)$, $r_2=\dim (\mathscr{R}_N)$, then $\alpha$ has $r_1$ elements, and $\beta$ has $r_2$ elements. Thus, if such a nice column selection $\beta$ (respectively nice row selection $\alpha$) has $k$ elements, then $\beta \subseteq (Q^{\le k-1} \times Q \times J) \cup Q^{\le k-1}$ (respectively $\alpha \subseteq Q^{\le k-1} \times Q \times I$).

We can now formulate the following extension of the method in the previous section in terms of nice selections. To this end, we extend the notion of a partial realization as follows.

\begin{Definition} \label{def:alphabeta_partreal}
	Let $\alpha$ be a nice row selection, and $\beta$ be a nice column selection of an LSS $\Sigma$ of the form \eqref{LSS}. Let
	$\widetilde{B}=\begin{bmatrix} x_0 & B_1 & \ldots & B_{D} \end{bmatrix}$,
	$\widetilde{C}=\begin{bmatrix} C^{\mathrm{T}}_1 & \ldots & C^{\mathrm{T}}_{D} \end{bmatrix}^{\mathrm{T}}$,
	
	\begin{enumerate}
		\item $\Sigma$ is a $\alpha$-partial realization of $f$, if for every $(v,q,i) \in \alpha$, the $\left[ p(q-1)+i \right]$th row of $M^f(v)$ equals the $i$th row of $C_qA_v\widetilde{B}$. This can be formulated equivalently as
		\[
		e_i^{\mathrm{T}}C_qA_v\widetilde{B}=M^f(v)_{p(q-1)+i,:}, \forall (v,q,i) \in \alpha.
		\]
		where $e_i$ denotes the $i$th unit vector in the canonical basis for $\mathbb{R}^p$.
		\item $\Sigma$ is a $\beta$-partial realization of $f$, if for every $(w,q,j) \in \beta$, the $\left[ m(q-1)+j+1 \right]$th column of $M^f(w)$ equals the $j$th column of $\widetilde{C}A_vB_q$, and if for every $w \in \beta$, $1$st column of $M^f(w)$ equals $\widetilde{C}A_vx_0$. This can be formulated equivalently as
		\begin{equation*}
		\begin{aligned}
		& \widetilde{C}A_wB_qe_{j}=M^f(w)_{:,m(q-1)+j+1}, \forall (w,q,j) \in \beta \\
		& \widetilde{C}A_vx_0 = M^f(w)_{:,1} \forall w \in \beta \\
		\end{aligned}
		\end{equation*}
		where $e_{j}$ denotes the $j$th unit vector in the canonical basis for $\mathbb{R}^{m}$.
		\item $\Sigma$ is an $(\alpha,\beta)$ partial realization of $f$, if for every $(v,q,i) \in \alpha$ and for every $(w,q,j) \in \beta$, the entry of $M^f(vw)$ in its $\left[ p(q-1)+i \right]$th row and $\left[ m(q-1)+j+1 \right]$th column equals the $(i,j)$th entry of $C_qA_{wv}B_q$, and if for every $(v,q,i) \in \alpha$ and for every $w \in \beta$, the entry of $M^f(vw)$ in its $\left[ p(q-1)+i \right]$th row and $1$st column equals the $i$th row of $C_qA_{wv}x_0$; alternatively,
		\begin{align*}
			& M^f(vw)_{p(q-1)+i , m(q-1)+j+1} = e_i^{\mathrm{T}}C_qA_{wv}B_qe_j, \\
			& \forall (v,q,i) \in \alpha, (w,q,j) \in \beta; \\
			& M^f(vw)_{p(q-1)+i , 1} = e_i^{\mathrm{T}}C_qA_{wv}x_0,  \\
			& \forall (v,q,i) \in \alpha, w \in \beta.
		\end{align*}
	\end{enumerate}
	
\end{Definition}

Note that the same definition could have been formulated for the arbitrary sets $\alpha \in (Q^* \times Q \times I)$ and $\beta \in (Q^* \times Q \times J) \cup Q^*$ which are not necessarily nice selections. However, Definition \ref{def:alphabeta_partreal} formulated as it is, since the algorithms (which will be presented later on) to acquire $\alpha$, $\beta$ or $(\alpha,\beta)$-partial realizations make use of Theorems \ref{theo:krylov1}-\ref{theo:krylov3}, and for the proof of these theorems, it is crucial that the sets $\alpha$ and $\beta$ define nice selections. This fact is also required for proving Theorem \ref{thm:nice_sel_switch_seq}, which gives the conditions for acquiring a reduced order LSS which has exactly the same input-output behavior as the original one, for a specific switching sequence.

\begin{Theorem} \label{theo:krylov1}
	\emph{(One sided moment matching by the column nice selection $\beta$).}
	Let $\Sigma$ be a realization of $f$ of the form \eqref{LSS}. In addition, let $V \in \mathbb{R}^{n \times r}$ be a full column rank matrix and $\beta$ be a nice column selection such that
	\[
	\mathscr{R}_{\beta} = \IM (V).
	\]
	For each $q \in Q$, define
	\[
	\begin{split}
	\bar{A}_q=V^{-1}A_qV \mbox{, } \bar{C}_q=C_qV \mbox{, } \bar{B}_q=V^{-1}B_q \mbox{, } \bar{x}_0=V^{-1}x_0
	\end{split}
	\]
	where $V^{-1}$ is any left inverse of $V$.
	Then
	\[ \bar{\Sigma}=(n,m,p,Q,\{(\bar{A}_q,\bar{B}_q,\bar{C}_q)\}_{q \in Q},\bar{x}_0) \]
	is a $\beta$-partial realization of $f$.
\end{Theorem}

The theorem above is similar to Theorem \ref{theo:mert1}. The numerical task is again to compute a matrix $V$ such that $\IM (V)=\mathscr{R}_{\beta}$ in an efficient way. In the model reduction method to be presented, the solution for this task will be explained more in detail.

\begin{proof}\emph{(Theorem \ref{theo:krylov1}).}
	We first show that for any $(w,q,j) \in \beta$,
	
	\begin{equation}  \label{theo:krylob1:eq1}
		V\bar{A}_{w}\bar{B}_qe_j = A_{w}B_{q}e_j.
	\end{equation}
	
	The proof is by induction on the length of $w$. For $w=\epsilon$, $B_qe_j$ is a column of $\mathscr{R}_{\beta}$, and since $ \mathscr{R}_{\beta} = \IM (V)$, $B_{q}e_{j}=Vx_1$ for some $x_1 \in \mathbb{R}^r$. Notice that $VV^{-1}V=V$ , and hence $V\bar{B}_qe_j=VV^{-1}Vx_1=B_{q}e_j$. Assume the claim holds for all $(w,q,j) \in \beta$, $|w| \le k$. Let $(w,q,j) \in \beta$ be such that $|w|=k+1$, $w=u\hat{q}$, $u \in Q^{*}$, $\hat{q} \in Q$. Then from the properties of a nice selection it follows that $(u,q,j) \in \beta$, and hence by the induction hypothesis,
	\[
	V\bar{A}_u\bar{B}_qe_j=A_{u}B_{q}e_j.
	\]
	It then follows that
	\begin{align*}
		& \bar{A}_{\hat{q}}\bar{A}_u\bar{B}_qe_j = V^{-1}A_{\hat{q}}(V\bar{A}_u\bar{B}_qe_j) = \\
		& V^{-1}A_{\hat{q}}A_uB_qe_j = V^{-1}A_wB_qe_j
	\end{align*}
	Notice that from $(w,q,j) \in \beta$ it follows that $A_{w}B_qe_j \in \IM(\mathscr{R}_{\beta}) = \IM (V)$, and hence there exists $x_2 \in \mathbb{R}^{r}$ such that $A_{w}B_qe_j = Vx_2$. It then follows that
	\begin{align*}
		& V\bar{A}_{\hat{q}}\bar{A}_u\bar{B}_qe_j = VV^{-1}Vx_2 = \\
		& Vx_2=A_{w}B_{q}e_j.
	\end{align*}
	That is, we have shown that \eqref{theo:krylob1:eq1} holds.
	From \eqref{theo:krylob1:eq1} it follows that
	\[
	\forall q \in Q: C_qA_{w}B_qe_j=C_qV\bar{A}_w\bar{B}_qe_j=\bar{C}_q\bar{A}_w\bar{B}_qe_j.
	\]
	Similarly, we can show that $\bar{C}_q\bar{A}_w\bar{x}_0=C_qA_wx_0$ for all $w \in \beta$ i.e., $\bar{\Sigma}$ is a $\beta$-partial realization of $f$.
\end{proof}

By duality, we could formulate nice row selections, and also a two sided Krylov subspace projection method, as demonstrated in Theorems \ref{theo:krylov2} and \ref{theo:krylov3}.

\begin{Theorem} \label{theo:krylov2}
	\emph{(One sided moment matching by the row nice selection $\alpha$).}
	Let $\Sigma$ be a realization of $f$ of the form \eqref{LSS}. In addition, let $W \in \mathbb{R}^{r \times n}$ be a full row rank matrix and $\alpha \subseteq Q^* \times Q \times I$ be a nice row selection such that
	\[
	\mathscr{O}_\alpha = \ker (W).
	\]
	For each $q \in Q$, define
	\[
	\begin{split}
	\bar{A}_q=WA_qW^{-1} \mbox{, } \bar{C}_q=C_qW^{-1} \mbox{, } \bar{B}_q=WB_q \mbox{, } \bar{x}_0=Wx_0
	\end{split}
	\]
	where $W^{-1}$ is any right inverse of $W$. Then
	\[
	\bar{\Sigma}=(n,m,p,Q,\{(\bar{A}_q,\bar{B}_q,\bar{C}_q)\}_{q \in Q},\bar{x}_0)
	\]
	is an $\alpha$-partial realization of $f$.
\end{Theorem}

\begin{proof}\emph{(Theorem \ref{theo:krylov2}).}
	This follows from Theorem \ref{theo:krylov1} by duality. Consider a $W \in \mathbb{R}^{r \times n}$ which satisfies the assumption of the theorem. Recall that $\widetilde{\mathscr{O}}_\alpha = \SPAN \{ (e_i^\mathrm{T}C_qA_v)^\mathrm{T} \mid (v,q,i) \in \alpha \}$. Then, it is easy to see that $\widetilde{\mathscr{O}}_\alpha = \IM (W^{\mathrm{T}})$. We will show that for any $(v,q,i) \in \alpha$,
	\begin{equation} \label{eq:proof_krylov2_1}
		e_i^{\mathrm{T}}\bar{C}_q\bar{A}_vW=e_i^{\mathrm{T}}C_qA_v
	\end{equation}
	The proof is again by induction on the length of $v$. For $v=\varepsilon$, $(e_i^{\mathrm{T}}C_q)^\mathrm{T}$ belongs to $\widetilde{\mathscr{O}}_{\alpha}$ and hence, $e_i^{\mathrm{T}}C_q=x_1^{\mathrm{T}}W$ for some $x_1 \in \mathbb{R}^r$. Notice that $WW^{-1}W=W$ hence $e_i^{\mathrm{T}}\bar{C}_qW=x_1^{\mathrm{T}}WW^{-1}W=e_i^{\mathrm{T}}C_q$. Assume the claim holds for all $(v,q,i) \in \alpha$, $|v| \leq k$. Let $(v,q,i) \in \alpha$ be such that $|v|=k+1$, $v=\hat{q}u$, $u \in Q^*$, $\hat{q} \in Q$. Then from the properties of a nice selection it follows that $(u,q,i) \in \alpha$, and hence, by the induction hypothesis
	\[ e_i^{\mathrm{T}}\bar{C}_q\bar{A}_uW=e_i^{\mathrm{T}}C_qA_u. \]
	It then follows that
	\begin{equation} \label{eq:proof_krylov2_2}
		\begin{aligned} 
			& e_i^{\mathrm{T}}\bar{C}_q\bar{A}_u\bar{A}_{\hat{q}}=(e_i^{\mathrm{T}}\bar{C}_q\bar{A}_uW)A_{\hat{q}}W^{-1}= \\
			& e_i^{\mathrm{T}}C_qA_uA_{\hat{q}}W^{-1}=e_i^{\mathrm{T}}C_qA_vW^{-1}
		\end{aligned}
	\end{equation}
	Notice that from $(v,q,i) \in \alpha$ it follows that $(e_i^{\mathrm{T}}C_qA_v)^\mathrm{T} \in \widetilde{\mathscr{O}}_{\alpha} = \IM (W^\mathrm{T})$ and hence there exists $x_2 \in \mathbb{R}^r$ such that
	\begin{equation} \label{eq:proof_krylov2_3}
		e_i^{\mathrm{T}}C_qA_v=x_2^{\mathrm{T}}W.
	\end{equation}
	It then follows from \eqref{eq:proof_krylov2_2} and \eqref{eq:proof_krylov2_3} that
	\[ e_i^{\mathrm{T}}\bar{C}_q\bar{A}_u\bar{A}_{\hat{q}}W=x_2^{\mathrm{T}}WW^{-1}W=x_2^{\mathrm{T}}W=e_i^{\mathrm{T}}C_qA_v \]
	That is, we have shown that \eqref{eq:proof_krylov2_1} holds. From \eqref{eq:proof_krylov2_1} it follows that
	\[ \forall q \in Q: e_i^{\mathrm{T}}C_qA_vB_q=e_i^{\mathrm{T}}\bar{C}_q\bar{A}_vWB_q=e_i^{\mathrm{T}}\bar{C}_q\bar{A}_v\bar{B}_q, \]
	Similarly, we can show that $\bar{C}_q\bar{A}_v\bar{x}_0=C_qA_vx_0$ for all $(v,q,i) \in \alpha$ i.e., $\bar{\Sigma}$ is an $\alpha$-partial realization of $f$.
\end{proof}

\begin{Theorem} \label{theo:krylov3}
	\emph{(Two sided moment matching by row/column nice selections $\alpha$ and $\beta$).}
	Let $\Sigma$ be a realization of $f$ of the form \eqref{LSS}. Let $\alpha \subseteq Q^* \times Q \times I$ be a nice row selection, $\beta \subseteq Q^* \times Q \times J$ be a nice column selection. Let $W \in \mathbb{R}^{r \times n}$ be a full row rank matrix and $V \in \mathbb{R}^{n \times r}$ be a full column rank matrix such that
	\begin{enumerate}
		\item $\Rank (WV)=r$,
		\item $\mathscr{O}_{\alpha} = \ker (W)$,
		\item $\mathscr{R}_{\beta} = \IM (V)$.
	\end{enumerate}
	For all $q \in Q$, define
	\[
	\begin{split}
	\bar{A}_q=WA_qV(WV)^{-1} \mbox{, } \bar{C}_q=C_qV(WV)^{-1} \mbox{, } \bar{B}_q=WB_q \mbox{, } \bar{x}_0=Wx_0
	\end{split}
	\]
	Then
	\[ \bar{\Sigma}=(n,m,p,Q,\{(\bar{A}_q,\bar{B}_q,\bar{C}_q)\}_{q \in Q},\bar{x}_0) \]
	is an $(\alpha,\beta)$-partial realization of $f$, and it is also an $\alpha$ and $\beta$-partial realization of $f$.
\end{Theorem}

\begin{proof}\emph{(Theorem \ref{theo:krylov3}).}
	Define $P=V(WV)^{-1}$. Notice that the conditions of the theorem imply $WV$ is nonsingular so its inverse $(WV)^{-1}$ exists. Notice that since $P$ is again $n \times r$ with full column rank, a left inverse $P^{-1}$ of $P$ can be defined and $P$ is a right inverse of $W$. It then follows from Theorem~\ref{theo:krylov1} and Theorem~\ref{theo:krylov2} that $\bar{\Sigma}$ is an $\alpha$- and $\beta$-partial realization of $f$. More clearly, from the proof of Theorem \ref{theo:krylov1}, \eqref{theo:krylob1:eq1}, it follows that for any $(w,q,j) \in \beta$,
	\begin{equation} \label{theo:krylov3:eq1} 
		P\bar{A}_w\bar{B}_qe_j=A_wB_qe_j
	\end{equation}
	Using duality or from the proof of Theorem~\ref{theo:krylov2}, \eqref{eq:proof_krylov2_1}, it can be shown that for any $(v,q,i) \in \alpha$,
	\begin{equation} \label{theo:krylov3:eq2}
		e_i^{\mathrm{T}}\bar{C}_q\bar{A}_v W = e_i^{\mathrm{T}}C_qA_v.
	\end{equation}
	Notice that $WP=I_r$ and hence, combining \eqref{theo:krylov3:eq1} and \eqref{theo:krylov3:eq2} implies
	\begin{align*}
		& e_i^{\mathrm{T}}\bar{C}_q\bar{A}_v\bar{A}_w\bar{B}_qe_j = e_i^{\mathrm{T}}\bar{C}_q\bar{A}_vWP\bar{A}_w\bar{B}_qe_j = \\
		& e_i^{\mathrm{T}}C_qA_vA_wB_qe_j.
	\end{align*}
	The part about $x_0$ can be proven similarly. It follows that $\bar{\Sigma}$ is an $(\alpha,\beta)$ partial realization of $f$. 
\end{proof}

\begin{Remark}
	\emph{(Relation between $N$, $2N$ and $\alpha$, $\beta$, $(\alpha,\beta)$-partial realizations).}
	Note that if the $V$ matrix in Theorem~\ref{theo:krylov1} is such that $\IM(V)= \mathscr{R}_N$ and $W$ matrix in Theorem~\ref{theo:krylov2} is such that $ker(W)=\mathscr{O}_N$, then the acquired reduced order systems would be $N$-partial realizations for each case. Likewise, if the $V$ and $W$ matrices in Theorem~\ref{theo:krylov3} can be found such that $\IM(V)= \mathscr{R}_N$, $ker(W)=\mathscr{O}_N$ and $\Rank(WV)=r$, then the acquired reduced order system would be a $2N$-partial realization. In this sense, the method given in this section is a generalization of the previous method. In other words, $N$ or $2N$-partial realizations are just $\beta$, $\alpha$ or $(\alpha,\beta)$-partial realizations for a specific choice of $\alpha$, $\beta$ or $(\alpha,\beta)$. These choices would be in the following form: The set $\alpha$ contains all the elements of the form $(v,q,i) \in (Q^{\leq N} \times Q \times I)$; the set $\beta$ contains all the elements of the form $(w,q,j) \in (Q^{\leq N} \times Q \times J)$ and $w \in Q^{\leq N}$.
\end{Remark}

Now we will present three efficient algorithms of model reduction by moment matching, which compute either an $\alpha$, $\beta$, $(\alpha,\beta)$-partial realization $\bar{\Sigma}$ for an $f$ which is realized by an LSS $\Sigma$. Firstly, we present algorithms for computing some subspaces of $\mathscr{R}_{\beta}$ and $\mathscr{O}_{\alpha}$. Then, those algorithms will be used to acquire the matrices $V$ and $W$ in the Theorems \ref{theo:krylov1}-\ref{theo:krylov3} and hence, to formulate a global model reduction by moment matching method for LSSs.

\begin{Definition} \label{def:subsets}
	\emph{(The languages related to $\beta$ and $\alpha$).}
	Let $\beta$ be a column nice selection and  $\mathbb{J}_\beta= \{ (q,j) \in Q \times J \mid \exists w \in Q^* \mbox{ such that } (w,q,j) \in \beta \}$. Define the corresponding languages related to $\beta$ as
	\begin{align} \label{eq:reach_language}
		& L^\beta_0= \{ w \in Q^* \mid w \in \beta \} \\
		& L^\beta_{q,j} = \{ w \in Q^* \mid (w,q,j) \in \beta \} \mbox{, } \forall (q,j) \in \mathbb{J}_\beta.
	\end{align}
	Furthermore, let $\alpha$ be a row nice selection and $\mathbb{I}_\alpha= \{ (q,i) \in Q \times I \mid \exists v \in Q^* \mbox{ such that } (v,q,i) \in \alpha \}$. Define the corresponding languages related to $\alpha$ as
	\begin{equation} \label{eq:obs_language}
		L^\alpha_{q,i}= \{ v \in Q^* \mid (v,q,i) \in \alpha \} \mbox{, } \forall (q,i) \in \mathbb{I}_\alpha.
	\end{equation}
	The numbers $t_\beta=|\mathbb{J}_\beta|+1$ and $t_\alpha=|\mathbb{I}_\alpha|$ will be called the \emph{subset cardinality} of $\beta$ and $\alpha$ respectively.
\end{Definition}

\begin{Example}
	Suppose a column nice selection $\beta$ related to an LSS $\Sigma=(1,1,n,Q,\{(A_q,B_q,C_q)|q \in Q\},x_0)$ with $Q=\{1,2,3,4\}$ is given by
	\begin{align*}
		\beta= & \{ \varepsilon, 2, (\varepsilon,1,1), (1,1,1), (3,1,1), (32,1,1), (34,1,1), \\
		& (\varepsilon,3,1), (\varepsilon,4,1), (1,4,1) \}.
	\end{align*}
	Then the set $\mathbb{J}_\beta$ and the corresponding languages $L^\beta_0, L^\beta_{1,1}, L^\beta_{3,1}, L^\beta_{4,1}$ are given by
	\begin{align*}
	    & \mathbb{J}_\beta=\{ (1,1), (3,1), (4,1) \} \\
		& L^\beta_0= \{ \varepsilon, 2 \} \\
		& L^\beta_{1,1}= \{ \varepsilon, 1, 3, 32, 34 \} \\
		& L^\beta_{3,1}= \{ \varepsilon \} \\
		& L^\beta_{4,1}= \{ \varepsilon, 1 \}.
	\end{align*}
	Note that the number $t_\beta=|\mathbb{J}_\beta|+1=4$ is the subset cardinality of $\beta$ i.e., there are $4$ languages related to $\beta$, namely $L_0, L_{1,1}, L_{3,1}$ and $L_{4,1}$.
\end{Example}

\begin{Definition} \label{def:NDFA}
	A non-deterministic finite state automaton (NDFA) is a tuple $\mathcal{A}=(S,Q,\{\rightarrow_q\}_{q \in Q} ,F, s_0)$ such that
	\begin{enumerate}
		\item $S$ is the finite state set,
		\item $F \subseteq S$ is the set of accepting (final) states,
		\item $\rightarrow_q \subseteq S \times S$ is the state transition relation labelled by $q \in Q$,
		\item $s_0 \in S$ is the initial state.
	\end{enumerate}
	For every $v \in Q^{*}$, define $\rightarrow_v$ inductively as follows: $\rightarrow_{\epsilon}=\{ (s,s) \mid s \in S \}$ and $\rightarrow_{vq} = \{ (s_1,s_2) \in S \times S \mid \exists s_3 \in S: (s_1,s_3) \in \rightarrow_{v} \mbox{ and } (s_3,s_2) \in \rightarrow_q \}$ for all $q \in Q$. We denote the fact $(s_1,s_2) \in \rightarrow_v$ by $s_1 \rightarrow_v s_2$. The fact that there exists $s_2$ such that $s_1 \rightarrow_v s_2$ is denoted by $s_1 \rightarrow_v$. Define the language $L(\mathcal{A})$ accepted by $\mathcal{A}$ as 
	\[
	L(\mathcal{A})=\{ v \in Q^{*} \mid \exists s \in F: s_0 \rightarrow_v s \}.
	\]
	We say that $\mathcal{A}$ is \emph{co-reachable}, if from any state a final state can be reached, i.e., for any $s \in S$, there exists $v \in Q^*$ and $s_f \in F$ such that $s \rightarrow_v s_f$. It is well-known that if $\mathcal{A}$ accepts $L$, then we can always compute an NDFA $\mathcal{A}_{co-r}$ from $\mathcal{A}$ such that $\mathcal{A}_{co-r}$ accepts $L$ and it is co-reachable.
\end{Definition}

In the sequel, we will assume that the languages $L^\beta_0$, $L^\beta_{q,j}$, $L^\alpha_{q,i}$ associated with a nice selection $\beta$ or $\alpha$ are regular i.e., there exists an NDFA accepting them. By using the definitions above  we can define the subspaces $\mathscr{R}_{L,j}(G)$ and $\mathscr{O}_{L,i}(H)$ for real matrices $G$ and $H$ as
\begin{align*}
	& \mathscr{R}_{L,j}(G)= \SPAN \{ A_vG_{:,j} \mid v \in L \} \\
	& \mathscr{O}_{L,i}(H)= \bigcap_{v \in L} \ker (H_{i,:}A_v)
\end{align*}
and use them to rewrite the spaces $\mathscr{R}_{\beta}$ and $\mathscr{O}_{\alpha}$ in the following form:
\begin{align} \label{eq:new_Rbeta_Oalpha}
	& \mathscr{R}_{\beta}= \mathscr{R}^{\beta}_{L_0,1}(x_0) + \sum\limits_{(q,j) \in \mathbb{J}_\beta}^{}\mathscr{R}^{\beta}_{(L_{q,j}),j}(B_q) \\
	& \mathscr{O}_{\alpha} = \bigcap_{(q,i) \in \mathbb{I}_\alpha} \mathscr{O}^{\alpha}_{(L_{q,i}),i}(C_q).
\end{align}

Now we are ready to present the two algorithms to compute a representation for the subspaces $\mathscr{R}_{L,j}(G)$ and $\mathscr{O}_{L,i}(H)$ respectively. Observe from \eqref{eq:new_Rbeta_Oalpha}, those two algorithms can be subsequently used for computing the $V$ and $W$ matrices such that $\IM (V) = \mathscr{R}_{\beta}$ and $\ker(W) = \mathscr{O}_{\alpha}$ for a given $\beta$ or $\alpha$. These algorithms are similar to the ones in \cite{bastugCDC2014} where they were used for model reduction of a discrete time LSS with respect to a certain set of switching sequences.

\begin{algorithm}
	\caption{
		Calculate  a matrix representation of $\mathscr{R}_{K,j}(G)$, 
		\newline
		\textbf{Inputs}: $(\{A_q\}_{q \in Q},G)$ and $\hat{\mathcal{A}}=(S,\{\rightarrow_q \}_{q \in Q},F,s_0)$ such that $L(\hat{\mathcal{A}})=K$, $j \in J$, $F=\{s_{f_1},\dots s_{f_k}\}$, $k \geq 1$ and $\hat{\mathcal{A}}$  is co-reachable.
		\newline
		\textbf{Outputs:} ${V}  \in \mathbb{R}^{n \times \hat{r}}$ such that $\Rank(\hat{V})=\hat{r}$,
		$\IM (\hat{V}) = \mathscr{R}_{K,j}(G)$. 
	}
	\label{alg4}
	\begin{algorithmic}[1]
		\STATE $\forall s \in S \backslash \{s_0\}: V_s:=0$.
		\STATE $V_{s_0}:=\mathbf{orth}(G_{:,j})$.
		\label{alg4.0}
		\STATE $\mathrm{flag}=0$.
		\WHILE{$\mathrm{flag}=0$}
		\label{alg4.1}
		\STATE $\forall s \in S: V_s^{old} := V_s$
		\FOR{$s \in S$}
		\STATE  $M_s:=V_s$
		\FOR{$q \in Q, s^{'} \in S: s^{'} \rightarrow_q s$}
		\STATE $M_s:=\begin{bmatrix} M_s, & A_qV^{old}_{s^{'}} \end{bmatrix}$
		\ENDFOR
		\STATE $V_s := \mathbf{orth}(M_s)$
		\ENDFOR
		\IF{$\forall s \in S: \Rank (V_s) = \Rank (V^{old}_s$)}
		\STATE{$\mathrm{flag}=1$.}
		\ENDIF 
		\ENDWHILE
		\RETURN $\hat{V}=\mathbf{orth} \left( \begin{bmatrix} V_{s_{f_1}} & \cdots & V_{s_{f_k}} \end{bmatrix} \right)$.
	\end{algorithmic}
\end{algorithm}

\begin{Lemma}[Correctness of Algorithm \ref{alg4} -- Algorithm \ref{alg5}] \label{lem:correctness}
	Assume $L$ is regular and $\hat{\mathcal{A}}$ is a co-reachable NDFA which accepts $L$. Algorithm \ref{alg4} returns a full column rank matrix $V$ such that $\IM (V)= \mathscr{R}_{L,j}(G)$, and Algorithm \ref{alg5} returns a full row rank matrix $W$ such that $\ker(W) = \mathscr{O}_{L,i}(H)$.
\end{Lemma}

\begin{proof}\emph{(Lemma \ref{lem:correctness}).}
	We prove only the first statement of the lemma, the second one can be shown using duality. Let $V_{s,i}=\SPAN\{ \IM (A_vG_{:,j}) \mid v \in Q^*, |v| \le i, s_0 \rightarrow_v s\}$, $i \in \mathbb{N}$. It then follows that after the execution of Step \ref{alg4.0}, $\IM (V_s)=V_{s,0}$ for all $s \in S$. Moreover, by induction it follows that
	\[
	V_{s,i+1}=V_{s,i} + \sum_{q \in Q, s^{'} \in S, s^{'} \rightarrow_q s} A_qV_{s^{'},i}
	\]
	for all $i=0,1,\ldots$ and $s \in S$. Hence, by induction it follows that at the $i$th iteration of the loop in Step \ref{alg4.1}, $\IM (V_s)=V_{s,i}$. Notice that $V_{s,i} \subseteq V_{s,i+1} \subseteq \mathbb{R}^n$ and hence there exists $k_s$ such that $V_{s,k_s}=V_{s,k}$, $k \ge k_s$, and thus $V_{s,k}=R_s$,
	\[
	R_s=\SPAN\{  A_vG_{:,j} \mid v \in Q^{*}, s_0 \rightarrow_v s \}.
	\]
	Let $k= \max \{k_s |  s \in S\}$. It then follows that $V_{s,k+1}=V_{s,k}=\IM (V_s)$ for all $s \in Q$ and hence after $k$ iterations, the loop \ref{alg4.1} will terminate. Moreover, in that case,  $\IM (V_{s_{f_i}})=R_{s_{f_i}}$, $i \in \{1,\cdots,k\}$. But notice that for any $v \in Q^*$, $q \in Q$, $s_0 \rightarrow_v s_{f_i}$ if and only if $v \in K$, and $s_0 \rightarrow_{qv} s_{f_i}$ if and only if $qv \in K$, $i \in \{1,\cdots,k\}$. Hence, $\sum\limits_{s \in F}^{}R_s=\mathscr{R}_{L,j}$ and thus $\IM \left( \left[ \begin{array}{ccc} V_{s_{f_1}} & \cdots & V_{s_{f_k}} \end{array} \right] \right)=\mathscr{R}_{L,j}$.
\end{proof}

\begin{algorithm}
	\caption{
		Calculate a matrix representation of $\mathscr{O}_{K,i}(H)$, 
		\newline
		\textbf{Inputs}: $(\{A_q\}_{q \in Q},H)$ and $\hat{\mathcal{A}}=(S,\{\rightarrow_q \}_{q \in Q},F,s_0)$ such that $L(\hat{\mathcal{A}})=K$, $i \in I$, $F=\{s_{f_1},\cdots s_{f_k}\}$, $k \geq 1$ and $\hat{\mathcal{A}}$  is co-reachable.
		\newline
		\textbf{Outputs:} $\hat{W}  \in \mathbb{R}^{\hat{r} \times n}$ such that $\Rank(\hat{W})=\hat{r}$,
		$\ker(\hat{W}) = \mathscr{O}_{K,i}(H)$.
	}
	\label{alg5}
	\begin{algorithmic}[1]
		\STATE $\forall s \in S \backslash F: W_s:=0$.
		\STATE $\forall s \in F: W_{s}^{\mathrm{T}}:=\mathbf{orth}(H_{i,:}^{\mathrm{T}})$.
		\STATE $\mathrm{flag}=0$.
		\label{alg2.0}
		\WHILE{$\mathrm{flag}=0$}
		\label{alg2.1}
		\STATE $\forall s \in S: W_s^{old} := W_s$
		\FOR{$s \in S$}
		\STATE  $M_s:=W_s$
		\FOR{$q \in Q, s^{'} \in S: s \rightarrow_q s^{'}$}
		\STATE $M_s:=\begin{bmatrix} M_s \\ W^{old}_{s^{'}}A_q \end{bmatrix}$
		\ENDFOR
		\STATE $W_s^{\mathrm{T}} := \mathbf{orth}(M_s^{\mathrm{T}})$
		\ENDFOR
		\IF{$\forall s \in S: \Rank (W_s) = \Rank (W^{old}_s)$}
		\STATE{$\mathrm{flag}=1$.}
		\ENDIF
		\ENDWHILE
		\RETURN $\hat{W}=W_{s_0}$.
	\end{algorithmic}
\end{algorithm}

Notice that the computational complexities of Algorithm \ref{alg4} and Algorithm \ref{alg5} are polynomial in $n$, even though the spaces of $\mathscr{R}_{L,j}(G)$ (resp. $\mathscr{O}_{L,i}(H)$) might be generated by images (resp. kernels) of exponentially many matrices.

Using Algorithms \ref{alg4} and \ref{alg5}, we can state Algorithms \ref{alg6}, \ref{alg7} and \ref{alg8} for getting reduced order $\alpha$, $\beta$ or $(\alpha,\beta)$ - partial realizations for an LSS $\Sigma$ respectively. The matrices $V$ and $W$ computed in Algorithms \ref{alg6} and \ref{alg7} satisfy the conditions of Theorems \ref{theo:krylov1} and \ref{theo:krylov2} respectively.

\begin{algorithm}
	\caption{Reduction for $\beta$-partial realization
		\newpage 
		\textbf{Inputs:} $\Sigma=(p,m,n,Q,\{(A_q,B_q,C_q)|q \in Q\},x_0)$, $\beta$ nice column selection, $\mathcal{A}^{\beta}_{0}$, $\mathcal{A}^{\beta}_{q,j}$ NDFAs such that $L(\mathcal{A}^{\beta}_{0})=L^\beta_0$ and $L(\mathcal{A}^{\beta}_{q,j})=L^\beta_{q,j}$ for all $(q,j) \in \mathbb{J}_\beta$.
		\newpage  
		\textbf{Output:} $\bar{\Sigma}=(p,m,r,Q,\{(\bar{A}_q,\bar{B}_q,\bar{C}_q)|q \in Q\},\bar{x}_0)$ such that $\bar{\Sigma}$ is a $\beta$-partial realization of $\Sigma$.
	}
	\label{alg6}
	\begin{algorithmic}[1]
		\STATE Use Algorithm \ref{alg4} with inputs $(\{A_q\}_{q \in Q},x_0)$, $j=1$ and NDFA $\mathcal{A}^\beta_0$. Store the output $\hat{V}$ as $V_{x_0}:=\hat{V}$.
		\FOR{$(q,j) \in \mathbb{J}_\beta$}
		\STATE  Use Algorithm \ref{alg4} with inputs $(\{A_q\}_{q \in Q},B_q)$, $j$ and NDFA $\mathcal{A}^\beta_{q,j}$.  Store the output $\hat{V}$ as $V_{q,j}:=\hat{V}$.
		\ENDFOR
		\STATE $V=\mathbf{orth}(\begin{bmatrix} V_{x_0} & V_1 & \cdots & V_{t_{\beta}-1} \end{bmatrix})$ where $t_\beta= |\mathbb{J}_\beta|+1$ is the subset cardinality of $\beta$ as in Def. \ref{def:subsets}.
		\STATE
		Let $r=\Rank (V)$, $V^{-1}$ be a left inverse of $V$ and set
		\[
		\bar{A}_q=V^{-1}A_qV \mbox{, } \bar{C}_q=C_qV \mbox{, } \bar{B}_q=V^{-1}B_q \mbox{, } \bar{x}_0=V^{-1}x_0.
		\]
		\RETURN $\bar{\Sigma}=(p,m,r,Q,\{(\bar{A}_q,\bar{B}_q,\bar{C}_q)|q \in Q\},\bar{x}_0)$.
    \end{algorithmic}
\end{algorithm}

\begin{algorithm}
	\caption{Reduction for $\alpha$-partial realization
		\newpage 
		\textbf{Inputs:} $\Sigma=(p,m,n,Q,\{(A_q,B_q,C_q)|q \in Q\},x_0)$, $\alpha$ nice row selection, $\mathcal{A}^{\alpha}_{q,i}$ NDFAs such that $L(\mathcal{A}^{\alpha}_{q,i})=L^\alpha_{q,i}$ for all $(q,i) \in \mathbb{I}_\alpha$.
		\newpage  
		\textbf{Output:} $\bar{\Sigma}=(p,m,r,Q,\{(\bar{A}_q,\bar{B}_q,\bar{C}_q)|q \in Q\},\bar{x}_0)$ such that $\bar{\Sigma}$ is an $\alpha$-partial realization of $\Sigma$.
	}
	\label{alg7}
	\begin{algorithmic}[1]
		\FOR{$(q,i) \in \mathbb{I}_\alpha$}
		\STATE  Use Algorithm \ref{alg5} with inputs $(\{A_q\}_{q \in Q},C_q)$, $i$ and NDFA $\mathcal{A}^\alpha_{q,i}$.  Store the output $\hat{W}$ as $W_{q,i}:=\hat{W}$.
		\ENDFOR
		\STATE $W^{\mathrm{T}}=\mathbf{orth}(\begin{bmatrix} W_{1}^{\mathrm{T}} & \cdots & W_{t_\alpha}^{\mathrm{T}} \end{bmatrix})$ where $t_\alpha=|\mathbb{I}_\alpha|$ is the subset cardinality of $\alpha$ as in Def. \ref{def:subsets}.
		Let $r=\Rank (W)$ and let $W^{-1}$ be a right inverse of $W$. Set
		\[
		\bar{A}_q=WA_qW^{-1} \mbox{, } \bar{C}_q=C_qW^{-1} \mbox{, } \bar{B}_q=WB_q \mbox{, } \bar{x}_0=Wx_0.
		\]
		\RETURN $\bar{\Sigma}=(p,m,r,Q,\{(\bar{A}_q,\bar{B}_q,\bar{C}_q)|q \in Q\},\bar{x}_0)$.
	\end{algorithmic}
\end{algorithm}

\begin{algorithm}
	\caption{Reduction for $(\alpha,\beta)$-partial realization
		\newpage 
		\textbf{Inputs:} $\Sigma=(p,m,n,Q,\{(A_q,B_q,C_q)|q \in Q\},x_0)$, $\beta$ nice column selection, $\alpha$ nice row selection, $\mathcal{A}^{\beta}_{0}$, $\mathcal{A}^{\beta}_{q,j}$, $\mathcal{A}^{\alpha}_{q,i}$ NDFAs such that $L(\mathcal{A}^{\beta}_{0})=L^\beta_0$ and $L(\mathcal{A}^{\beta}_{q,j})=L^\beta_{q,j}$ for all $(q,j) \in \mathbb{J}_\beta$ and $L(\mathcal{A}^{\alpha}_{q,i})=L^\alpha_{q,i}$ for all $(q,i) \in \mathbb{I}_\alpha$.
		\newpage  
		\textbf{Output:} $\bar{\Sigma}=(p,m,r,Q,\{(\bar{A}_q,\bar{B}_q,\bar{C}_q)|q \in Q\},\bar{x}_0)$ such that $\bar{\Sigma}$ is an $(\alpha,\beta)$-partial realization of $\Sigma$.
	}
	\label{alg8}
	\begin{algorithmic}[1]
		\STATE Compute the matrix $V$ as in Algorithm \ref{alg6}.
		\STATE Compute the matrix $W$ as in Algorithm \ref{alg7}.
		\IF{$\Rank(W)=r$ and $\Rank(V)=r$ and $\Rank(WV)=r$}
		\STATE $P=V(WV)^{-1}$
		\[
		\bar{A}_q=WA_qP \mbox{, } \bar{C}_q=C_qP \mbox{, } \bar{B}_q=WB_q \mbox{, } \bar{x}_0=Wx_0.
		\]
		\ENDIF
		\RETURN $\bar{\Sigma}=(p,m,r,Q,\{(\bar{A}_q,\bar{B}_q,\bar{C}_q)|q \in Q\},\bar{x}_0)$.
	\end{algorithmic}
\end{algorithm}

\begin{Lemma} [Correctness of Algorithms \ref{alg6}, \ref{alg7} and \ref{alg8}]
	Let $\Sigma$ be an LSS of the form \eqref{LSS}.
	\begin{enumerate}
		\item  Let $\beta$ be a nice column selection and assume that $L^\beta_0$, $L^{\beta}_{q,j}$, $(q,j) \in \mathbb{J}_\beta$ are regular languages. Let $\mathcal{A}^\beta_0$ and $\mathcal{A}^\beta_{q,j}$, $(q,j) \in \mathbb{J}_\beta$ be co-reachable NDFAs which accept $L^\beta_0$ and $L^{\beta}_{q,j}$, $(q,j) \in \mathbb{J}_\beta$ respectively. Then the LSS $\bar{\Sigma}$ returned by Algorithm \ref{alg6} is a $\beta$-partial realization of $f=Y_{\Sigma,x_0}$.
		
		\item Let $\Sigma$ be an LSS of the form \eqref{LSS}. Let $\alpha$ be a nice row selection and assume that $L^{\alpha}_{q,i}$, $(q,i) \in \mathbb{I}_\alpha$ are regular languages. Let $\mathcal{A}^\alpha_{q,i}$, $(q,i) \in \mathbb{I}_\alpha$ be co-reachable NDFAs which accept $L^{\alpha}_{q,i}$, $(q,i) \in \mathbb{I}_\alpha$ respectively. Then the LSS $\bar{\Sigma}$ returned by Algorithm \ref{alg7} is an $\alpha$-partial realization of $f=Y_{\Sigma,x_0}$.
		
		\item Let $\Sigma$ be an LSS of the form \eqref{LSS}. Let $\alpha$ be a nice row selection, $\beta$ be a nice column selection and assume that $L^\beta_0$, $L^{\beta}_{q,j}$, $(q,j) \in \mathbb{J}_\beta$, $L^{\alpha}_{q,i}$, $(q,i) \in \mathbb{I}_\alpha$ are regular languages. Let $\mathcal{A}^\beta_0$, $\mathcal{A}^\beta_{q,j}$, $(q,j) \in \mathbb{J}_\beta$ and $\mathcal{A}^\alpha_{q,i}$, $(q,i) \in \mathbb{I}_\alpha$ be co-reachable NDFAs which accept $L^\beta_0$, $L^{\beta}_{q,j}$, $(q,j) \in \mathbb{J}_\beta$, $L^{\alpha}_{q,i}$, $(q,i) \in \mathbb{I}_\alpha$ respectively. Then the LSS $\bar{\Sigma}$ returned by Algorithm \ref{alg8} is an $(\alpha,\beta)$-partial realization of $f=Y_{\Sigma,x_0}$ if the condition $\Rank(W)=r$ and $\Rank(V)=r$ and $\Rank(WV)=r$ holds.
	\end{enumerate}
\end{Lemma}

The final result of the paper will be to connect a certain switching sequence for a continuous time LSS to a nice column or row selection. For this, we need additional definitions as below.


\begin{Definition}[The generating language] \label{def:generating_lang}
	The generating language $\mathcal{L}$ for a sequence of discrete modes $\upsilon=q_1q_2 \cdots q_k \in Q^{*}$, $q_1,\ldots, q_k \in Q$, $k \ge 2$, is defined as
	\begin{equation}
		\mathcal{L}_{\upsilon}=\{ v \in Q^* \mid v=(q_1)^{\omega_1}(q_2)^{\omega_2} \cdots (q_k)^{\omega_k} \mbox{, } \omega_1, \dots, \omega_{k} \in \mathbb{N} \}
	\end{equation}
\end{Definition}

\begin{Definition} \label{nice_sel_switch_seq}
	\emph{(Nice selection related to a switching sequence).}
	A nice column selection $\beta_{\mu}$ related to a  sequence
	of discrete modes
	  $\upsilon=q_1q_2 \cdots q_k \in Q^{+}$, $q_1,\ldots, q_k \in Q$, $k \ge 2$ is defined as
	\begin{equation}
		\beta_{\upsilon}=\{ (w,q_0,j) \mid q_0w \in \mathcal{L}_{\upsilon}, j \in J, q_0 \in Q \} \cup \{ w \mid w \in \mathcal{L}_{\upsilon} \}
	\end{equation}
	In addition, a nice row selection $\alpha_{\upsilon}$ related to a sequence of discrete modes
	$\upsilon=q_1q_2 \cdots q_k \in Q^{+}$, $q_1,\ldots, q_k \in Q$, $k \ge 2$
	is defined as
	\begin{equation}
		\alpha_{\upsilon}=\{ (v,q,i) \mid vq \in \mathcal{L}_{\upsilon}, i \in I, q \in Q \}
	\end{equation}
\end{Definition}

The following theorem makes it possible to use the model reduction method with nice selections with respect to a specific switching sequence.

\begin{Theorem} \label{thm:nice_sel_switch_seq}
	Consider a sequence of discrete modes
	$\upsilon=q_1q_2 \cdots q_k \in Q^{+}$, $q_1,\ldots, q_k \in Q$, $k \ge 2$
	Let $\bar{\Sigma}$ be an LSS which is a $\beta_\upsilon$ (resp. $\alpha_\upsilon$) - partial realization of $f=Y_{\Sigma,x_0}$. Then, for every switching signal which satisfies \eqref{eq:switch_seq_sigma_def} for some $t_1,\ldots, t_k > 0$, and for all $u \in L_{loc}(\mathbb{R}_+,\mathbb{R}^m)$,
	  \[ \forall s \in [0,t]: Y_{\Sigma,x_0}(u,\sigma)(s)=Y_{\bar{\Sigma},\bar{x}_0}(u,\sigma)(s) \]
	  where $t=t_1+\cdots+t_k$.
\end{Theorem}

Intuitively, Theorem \ref{thm:nice_sel_switch_seq} says that if $\bar{\Sigma}$ is a $\beta_\upsilon$ (resp. $\alpha_\upsilon$) - partial realization of $f=Y_{\Sigma,x_0}$, then the outputs of $\Sigma$ and $\bar{\Sigma}$ along the switching sequence $\mu=(q_1,t_1) \dots (q_k,t_k)$ are the same. Hence, if we apply Algorithm \ref{alg6} or Algorithm \ref{alg7} with $\beta=\beta_\upsilon$ or respectively $\alpha=\alpha_\upsilon$, then we will get an LSS $\bar{\Sigma}$ which has the same input-output behavior as $\Sigma$ along the switching sequence $\mu$.

\begin{proof}\emph{(Theorem \ref{thm:nice_sel_switch_seq}).}
	Only the part related with the nice column selection will be proven, similar arguments can be used to prove the result for
	nice row selections.
	
	Note that the output of $\Sigma$ at a time $s \in [t_1+\cdots+t_{i-1}, t_1+\cdots+t_{i-1}+t_i]$, $i=2, \dots, k$ due to the switching signal $\sigma$, initial state $x_0$ and input $u \in L_{loc}(\mathbb{R}_+,\mathbb{R}^m) $ is given by
	\begin{equation} \label{eq:output_pf_thm9}
		\begin{aligned}
			& Y_{\Sigma,x_0}(u,\sigma)(s) =  C_{q_i}e^{A_{q_i}[s-(t_1+\cdots+t_{i-1})]}e^{ A_{q_{i-1}}t_{i-1} } \cdots e^{A_{q_1}t_1}x_0 \\
			& + \sum_{j=1}^{i-1} \int_{0}^{t_j}C_{q_{i-1}} e^{A_{q_{i-1}}t_{i-1}} \cdots e^{A_{q_j}(t_j-\tau)}B_{q_j} u \left( \tau + \sum_{l=1}^{j-1} t_l \right)d\tau \\
			& + \int_{0}^{s}C_{q_i} e^{A_{q_i}[s-(t_1+\cdots+t_{i-1})-\tau]}B_{q_i} u \left( \tau + \sum_{l=1}^{i-1} t_l \right)d\tau
		\end{aligned}
	\end{equation}
	whereas $Y_{\bar{\Sigma},\bar{x}_0}$ is given by
	\begin{equation} \label{eq:output_pf_thm9_2}
		\begin{aligned}
			& Y_{\bar{\Sigma},\bar{x}_0} (u,\sigma)(s) =  \bar{C}_{q_i}e^{\bar{A}_{q_i}[s-(t_1+\cdots+t_{i-1})]}e^{ \bar{A}_{q_{i-1}}t_{i-1} } \cdots e^{\bar{A}_{q_1}t_1}\bar{x}_0 \\
			& + \sum_{j=1}^{i-1} \int_{0}^{t_j}\bar{C}_{q_{i-1}} e^{\bar{A}_{q_{i-1}}t_{i-1}} \cdots e^{\bar{A}_{q_j}(t_j-\tau)}\bar{B}_{q_j} u \left( \tau + \sum_{l=1}^{j-1} t_l \right)d \tau \\
			& + \int_{0}^{s}\bar{C}_{q_i} e^{\bar{A}_{q_i}[s-(t_1+\cdots+t_{i-1})-\tau]}\bar{B}_{q_i} u \left( \tau + \sum_{l=1}^{i-1} t_l \right)d \tau.
		\end{aligned}
	\end{equation}
	Hence, for $Y_{\Sigma,x_0}$ and $Y_{\bar{\Sigma},\bar{x}_0}$ to be equal, it is sufficient that the following equations hold:
	\begin{equation} \label{eq:pf_thm9_state_trans}
		\begin{aligned}
			 & C_{q_i} e^{A_{q_i}[s-(t_1+\cdots+t_{i-1})]} e^{ A_{q_{i-1}}t_{i-1} }   \cdots e^{A_{q_1}t_1}x_0  = \\ 
			 & \bar{C}_{q_i}e^{\bar{A}_{q_i}[s-(t_1+\cdots+t_{i-1})]}e^{ \bar{A}_{q_{i-1}}t_{i-1} } \cdots e^{\bar{A}_{q_1}t_1}\bar{x}_0, \\
			& C_{q_{i-1}}e^{A_{q_{i-1}}t_{i-1}} \cdots e^{A_{q_1}t_1}B_{q_1} = \bar{C}_{q_{i-1}}e^{\bar{A}_{q_{i-1}}t_{i-1}} \cdots e^{\bar{A}_{q_1}t_1}\bar{B}_{q_1}, \\
			& \mbox{    } \vdots \\
			& C_{q_{i-1}}e^{A_{q_{i-1}}t_{i-1}}B_{q_{i-1}} = \bar{C}_{q_{i-1}}e^{\bar{A}_{q_{i-1}}t_{i-1}}\bar{B}_{q_{i-1}},  \\
			& C_{q_i}e^{A_{q_i}[s-(t_1+\cdots+t_{i-1})]}B_{q_i} = \bar{C}_{q_i}e^{\bar{A}_{q_i}[s-(t_1+\cdots+t_{i-1})]}\bar{B}_{q_i}.
		\end{aligned}
	\end{equation}
	By Definition \ref{def:generating_lang}, the generating language for $\upsilon$ can be defined as the set
	\begin{equation*}
		\mathcal{L}_{\upsilon}= \{ (q_1)^{\omega_1}(q_2)^{\omega_2}\cdots (q_k)^{\omega_k} \mid \omega_1, \dots, \omega_{k} \in \mathbb{N}\}.
	\end{equation*}
	Therefore, if the Taylor series expansion of the matrix exponentials in the equations of \eqref{eq:pf_thm9_state_trans} is taken around $t=0$, it can be seen that for
	\eqref{eq:pf_thm9_state_trans} to hold, it is sufficient that  
	\begin{equation} \label{eq:pf_thm9_Markov_eq}
		\begin{aligned}
			& C_qA_wx_0 = \bar{C}_q\bar{A}_w\bar{x}_0 \mbox{ for all } wq \in \mathcal{L}_{\upsilon} \\
			& C_qA_wB_{q_0} = \bar{C}_q\bar{A}_w\bar{B}_{q_0} \mbox{ for all } q_0wq \in \mathcal{L}_{\upsilon}.
		\end{aligned}
	\end{equation}
	holds. In turn, \eqref{eq:pf_thm9_Markov_eq} follows from the
	assumption that $\bar{\Sigma}$ is a $\beta_{\upsilon}$-partial
	realization of $Y_{\Sigma,x_0}$, if we use the definition of
	$\beta_{\upsilon}$.
	Hence, $Y_{\Sigma,x_0}(u,\sigma)(s)=Y_{\bar{\Sigma},\bar{x}_0}(u,\sigma)(s)$ for $s \in [0,t_1+t_2+\cdots+t_k]$.
\end{proof}

The theorem above builds the relationship between a certain switching sequence and its related nice selection. Hence, it makes it possible to acquire an approximation to an LSS whose input-output behavior is identical for all switching sequences $\mu=(q_1,t_1) \cdots (q_k,t_k)$ for a fixed sequence of discrete modes $q_1,\ldots,q_k$, and whose order $r$ is possibly smaller than $n$ (Note that since $V \in \mathbb{R}^{n \times r}$ is of full column rank, $r \leq n$).

\section{Numerical examples}
\label{sec:exam}

In this section, two generic numerical examples are presented to illustrate the model reduction procedure. One of the numerical examples is for an LSS who has stable local modes. With this example, it is aimed to show the flexibility of the nice selections about choosing the specific local modes, on which the approximation should focus. Whereas in the other numerical example, the LSS has unstable local modes, and an $N$-partial realization is acquired for the original system to illustrate a solution to the analogue of the moment matching problem for linear systems.

Firstly, the procedure is applied to a SISO, $11$th order LSS with $2$ discrete modes i.e., to an LSS $\Sigma$ of the form $\Sigma=(p,m,n,Q,\{(A_q,B_q,C_q)|q \in Q\},x_0)$ with $p=m=1$, $n=11$, $Q=\{1,2\}$. The randomly generated system has locally stable modes. The data of $A_q$, $B_q$, $C_q$ parameters and the initial state $x_0$ used for simulation is also available from https://kom.aau.dk/\texttildelow mertb/. A random switching signal with minimum dwell time (time between two subsequent changes in the switching signal) of $0.4$ for mode $1$ and $0.1$ for mode $2$ is used for simulation. Note that the minimum dwell time for the first mode is chosen to be higher since for this example, the approximation will be focused more on mode $1$ than mode $2$. The input $u(t)$ used for simulation is an array of white Gaussian noise. The simulation time interval used is $t=[0,1]$ \footnote{Recall that the method is based on matching the coefficients of the Taylor series expansion for $Y_{\Sigma,x_0}$ around $t=0$, hence the simulation time horizon should be chosen ``small enough''. It should be noted that coming up with a priori error bounds for the moment matching problem is challenging even for the linear case \cite{antoulas}. Consequently, the matter of up to which time horizon the method gives a ``good'' approximation is an open problem, and for now, it can be decided by a posteriori experiments related to the specific problem at hand.}. For the nice selection $\beta$ given as
\begin{align*}
	\beta= & \{ \varepsilon, 1, (\varepsilon,1,1), (1,1,1), (11,1,1), \\ 
	& (111,1,1), (112,1,1), (\varepsilon,2,1) \},
\end{align*}
an approximation LSS $\bar{\Sigma}$ of order $8$ is acquired which is a $\beta$-partial realization of $\Sigma$. Note that, from the set $\beta$, it can be seen that the approximation is desired to be focused more on mode $1$ than mode $2$. In Fig. \ref{fig:example1}, $6$ plots are shown for comparison of the outputs of $\Sigma$ and $\bar{\Sigma}$ for random switching sequences $\sigma(t)$ with given properties. It can be seen from Fig. \ref{fig:example1}, whenever the first operating mode is mode $1$ and mode $1$ operates much more than mode $2$ in total time horizon, the approximation is better. Last point to mention about this example is that the same simulation is ran for $500$ hundred times with random switching sequences with the given properties, for the case when the first operating mode is mode $1$. The best fit rates (BFRs) for each simulation is calculated according to the following (\cite{ljung}, \cite{toth2012})
\[
\mbox{BFR}=100 \%. \max \left( 1-\frac{\norm{y(\cdot)-\bar{y}(\cdot)}_2}{\norm{y(\cdot)-y_m}_2},0 \right)
\]
and mean of the BFRs over these $500$ simulations is acquired as $73.5848\%$, whereas the best and worst acquired BFR is $88.6476\%$ and $13.0214
\%$ respectively. The mean of BFR values over $500$ simulations for this example implies that the method yields a good approximation for such a system $\Sigma$ in the given time horizon.
\begin{figure}
	\centering
	\includegraphics[width=0.5\textwidth]{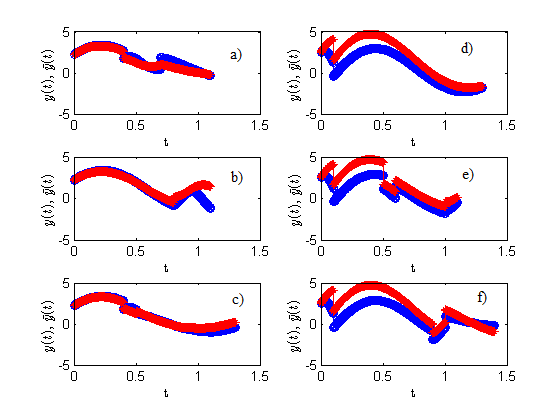}
	\caption{The response $y(t)$ (in blue) of the original LSS $\Sigma$ of order $11$ and the response $\bar{y}(t)$ (in red) of its $\beta$-partial realization $\bar{\Sigma}$ of order $8$ for various switching sequences. The switching sequences $\mu$ for each plot are as follows: $a)$ $\mu= (1,0.4)(2,0.3)(1,0.4)$, $b)$ $\mu=(1,0.8)(2,0.3)$, $c)$ $\mu=(1,0.4)(2,0.1)(1,0.8)$, $d)$ $\mu=(2,0.1)(1,1.2)$, $e)$ $\mu=(2,0.1)(1,0.4)(2,0.1)(1,0.4)(2,0.1)$, $f)$ $\mu=(2,0.1)(1,0.8)(2,0.1)(1,0.4)$}
	\label{fig:example1}
\end{figure}

The procedure is also applied to get a reduced order approximation to an LSS whose local modes are unstable. The original LSS used in this case is an LSS of the form $\Sigma=(p,m,n,Q,\{(A_q,B_q,C_q)|q \in Q\},x_0)$ with $p=m=1$, $n=12$ and $Q=\{1,2\}$. The resulting reduced order model $\bar{\Sigma}_1$ is a $1$-partial realization of $y_{\Sigma}$ of order $9$. Note that the precise number of matched Markov parameters of the form $M^f(v)$ is equal to the number of words in the set $Q^{\leq 1}$, and it is given by
\[
\frac{D^{N+1}-1}{D-1}=\frac{2^{1+1}-1}{2-1}=3.
\]
The same parameters in the first example are used with the exception of minimum dwell time for both modes being $0.1$ and the simulation time horizon being $t=[0,3]$. Again the output $y(t)$ of the original system $\Sigma$ and the output $\bar{y}(t)$ of the reduced order system $\bar{\Sigma}_1$ are simulated for $500$ random switching sequences and input trajectories. The mean of the BFRs for this example is $79.0518\%$; whereas, the best acquired BFR is $90.8013\%$ and the worst $62.7846\%$. The outputs $y(t)$ and $\bar{y}(t)$ of the most successful simulation for this example are illustrated in Fig.~\ref{fig:example2}.
\begin{figure}
	\centering
	\includegraphics[width=0.5\textwidth]{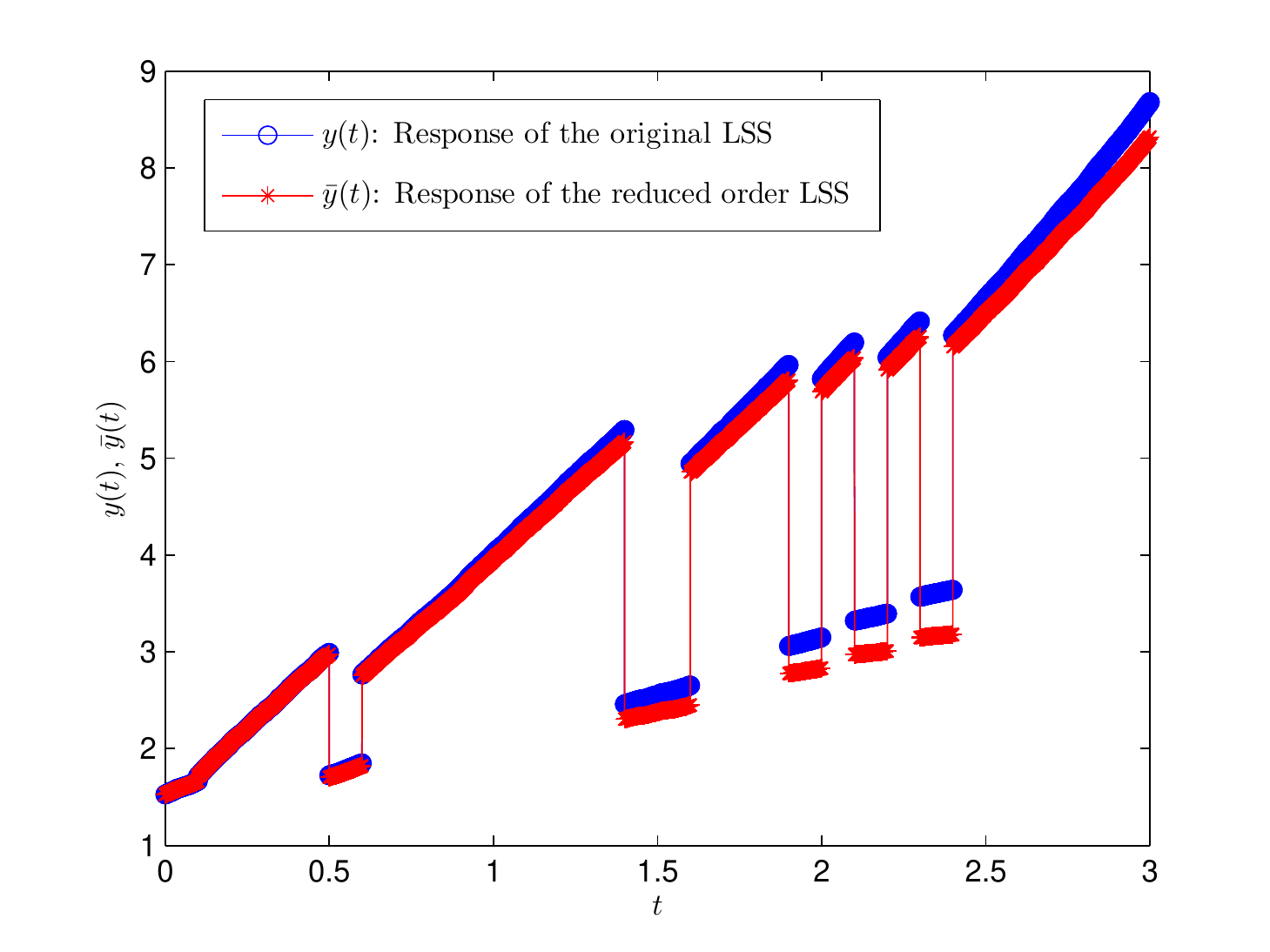}
	\caption{The response $y(t)$ of the original LSS $\Sigma$ of order $12$ and the response $\bar{y}(t)$ of the reduced order approximation LSS $\bar{\Sigma}_1$ of order $9$.}
	\label{fig:example2}
\end{figure}

\section{Conclusion}
\label{sec:conc}

Two moment matching procedures for model reduction of continuous time LSSs has been given. The first method is the direct analogue of the moment matching approaches in the linear case, for LSSs. The second method relies on the nice selections of some desired vectors in the reachability or observability space of an LSS. The notion of nice selections gives flexibility to the user of the procedure in the following sense: It is possible to focus the approximation on some preferred local modes more than the others. It has been proven that with this procedure, as long as a certain criterion is satisfied, it is possible to acquire at least one reduced order approximation to the original LSS whose Markov parameters related to the specific nice selection are matched with the original one's. Finally, it has been shown that nice selections can be used for matching the input - output behavior of an LSS with another one of possibly lower order, for a specific switching sequence. Discovering the relationship between a set of switching sequences with nice selections for continuous time LSSs would be a potential future research topic since it would solve the problem of approximation or minimization for restricted switching dynamics.


%

\appendices

\section{Proof of Theorem \ref{lem:nice_selections}}
To present the proof of Theorem \ref{lem:nice_selections} we will introduce an ordering on $Q^*$ as follows:

\begin{Definition} \label{def:lexico}
	\emph{(Ordering on $Q^*$).} Suppose that $Q=\{1,\dots,D\}$. Let the map $\phi: Q^* \rightarrow \mathbb{N}$ be defined as follows:
	\begin{equation}
		\begin{aligned}
			& \phi(\varepsilon)=0 \\
			& \phi(v)=q_1(D+1)^{k-1}+q_2(D+1)^{k-2} + \cdots +q_k.
		\end{aligned}
	\end{equation}
	where $v = q_1q_2 \cdots q_k$, $q_1, \dots, q_k \in Q$, $k \geq 1$.
	Then an \emph{ordering} $\prec$ on the elements of $Q^*$ can be defined as follows: For any two words $v,w \in Q^*$, if $\phi(v) < \phi(w)$, then $v \prec w$. 
\end{Definition}

Intuitively, this ordering states that $v \prec w$ if $w$ is bigger than $v$ when the words $v,w$ are interpreted as integer numbers in the basis $D+1$. Note that for any $v,w \in Q^{*}$,  $v \prec w$ implies $|v| \le |w|$, and $|v| < |w|$ implies  $v \prec w$.

\begin{proof}[Proof of Theorem \ref{lem:nice_selections}]
\emph{i)} Let $R_{n-1}$ denote the matrix
\[ R_{n-1}= \begin{bmatrix}
A_{v_1}\widetilde{B} & A_{v_2}\widetilde{B} & \cdots & A_{v_{M_{n-1}}}\widetilde{B}
\end{bmatrix} \]
where $M_{n-1}$ denotes the cardinality of the set $Q^{\leq n-1}$; $v_1,v_2, \dots, v_{M_{n-1}} \in Q^{\leq n-1}$ and $v_1 \prec v_2 \prec \cdots \prec v_{M_{n-1}}$ with respect to the ordering in Definition \ref{def:lexico}.

In this part of the proof, $x_0$ of $\Sigma$ is assumed to be zero for simplicity in notation, note that the proof can easily be modified for the case when $x_0$ is nonzero. Since $\Sigma$ is assumed to be minimal, for any $r<n$, there exists $r$ linearly independent columns of $R_{n-1}$. Suppose these columns are picked in the following manner: Scanning through the columns of $R_{n-1}$ from left to right, choose the first $r$ columns linearly independent from the preceding columns. Our claim is that, this method would yield a nice column selection. To prove the theorem, we claim that if $(A_{\sigma}A_wB_q)_{:,j}$ is an element of the selection defined, $(A_wB_q)_{:,j}$ must also be an element i.e., if $(w\sigma,q,j) \in \beta$, $(w,q,j) \in \beta$. We prove this claim by contradiction. Suppose the columns are chosen in this way and for a $q, \sigma \in Q$, $w \in Q^*$ and $j \in \{1,...,m\}$, the $j$th column of $A_\sigma A_wB_q$ is an element of this selection while the $j$th column of $A_wB_q$ is not. This means $(A_wB_q)_{:,j}$ is a linear combination of the columns of $R_{n-1}$ preceding it while $(A_\sigma A_wB_q)_{:,j}$ is not. Let $x_1,...,x_k$ denote the columns of $R_{n-1}$ which precede the column $(A_wB_q)_{:,j}$ and $x_1,...,x_h$ with $h \geq k$ denote the columns of $R_{n-1}$ which precede the column $(A_\sigma A_wB_q)_{:,j}$. Note that for some $c_1,...,c_k \in \mathbb{R}$
\[
(A_wB_q)_{:,j}=c_1x_1+...+c_kx_k
\]
Thus the column $(A_\sigma A_wB_q)_{:,j}$ can be written as
\begin{equation} \label{eq:proof_nice_lin_indep1}
\begin{aligned}
(A_\sigma A_wB_q)_{:,j}=A_\sigma(A_wB_q)_{:,j} & =A_\sigma(c_1x_1+...+c_kx_k) \\ 
& =c_1A_\sigma x_1+...+c_kA_\sigma x_k
\end{aligned}
\end{equation}
and since all the vectors $A_\sigma x_1,...,A_\sigma x_k$ precede the column $(A_\sigma A_wB_q)_{:,j}$, each of them can also be written as a linear combination of the columns $x_1,...,x_h$ which precede $(A_\sigma A_wB_q)_{:,j}$. That means for some $a_{st} \in \mathbb{R}$, $s=1,...,h$, $t=1,...,h$, \eqref{eq:proof_nice_lin_indep1} can be rewritten as
\begin{equation*}
\begin{aligned}
(A_\sigma A_wB_q)_{:,j} & = c_1A_\sigma x_1+...+c_kA_\sigma x_k \\
& =c_1(a_{11}x_1+...+a_{1h}x_h)+... \\
& +c_k(a_{h1}x_1+...+a_{h1}x_h)
\end{aligned}
\end{equation*}
i.e., the column $(A_\sigma A_wB_q)_{:,j}$ is a linear combination of its preceding columns. This contradicts our assumption and concludes the proof of the reachability part.

\emph{ii)} This part is the dual of part \emph{i}.
\end{proof}

\ifCLASSOPTIONcaptionsoff
  \newpage
\fi



%

\bibliographystyle{plain}
\bibliography{./bare_jrnl}

\begin{thebibliography}{10}

\bibitem{antoulas}
A.~C. Antoulas.
\newblock {\em Approximation of Large-Scale Dynamical Systems}.
\newblock SIAM, Philadelphia, PA, 2005.

\bibitem{arnold_games_2003}
A.~Arnold, A.~Vincent, and I.~Walukiewicz.
\newblock Games for synthesis of controllers with partial observation.
\newblock {\em Theoretical Computer Science}, 303(1):7--34, June 2003.

\bibitem{BilinearMomentMatching3}
Z.~Bai and D.~Skoogh.
\newblock A projection method for model reduction of bilinear dynamical
  systems.
\newblock {\em Linear Algebra and its Applications}, 415(2 - 3):406 -- 425,
  2006.

\bibitem{bastugCDC2014}
M.~Bastug, M.~Petreczky, R.~Wisniewski, and J.~Leth.
\newblock Model reduction of linear switched systems by restricting discrete
  dynamics.
\newblock In {\em Accepted for publication in Proc. of the IEEE Conference on
  Decision and Control (CDC)}, Los Angeles, CA, USA.

\bibitem{bastugACC2014}
M.~Bastug, M.~Petreczky, R.~Wisniewski, and J.~Leth.
\newblock Model reduction by moment matching for linear switched systems.
\newblock In {\em Proc. of the American Control Conference (ACC)}, pages 3942
  -- 3947, Portland, OR, USA, June 2014.

\bibitem{French1}
A.~Birouche, J~Guilet, B.~Mourillon, and M~Basset.
\newblock Gramian based approach to model order-reduction for discrete-time
  switched linear systems.
\newblock In {\em Proc. Mediterranean Conference on Control and Automation},
  2010.

\bibitem{Chahlaoui}
Y.~Chahlaoui.
\newblock Model reduction of hybrid switched systems.
\newblock In {\em Proceeding of the 4th Conference on Trends in Applied
  Mathematics in Tunisia, Algeria and Morocco, May 4-8, Kenitra, Morocco},
  2009.

\bibitem{BilinearMomentMatching5}
G.~M. Flagg.
\newblock {\em Interpolation Methods for the Model Reduction of Bilinear
  Systems}.
\newblock PhD thesis, Virginia Polytechnic Institute, 2012.

\bibitem{Lam1}
H.~Gao, J.~Lam, and C.~Wang.
\newblock Model simplification for switched hybrid systems.
\newblock {\em Systems \& Control Letters}, 55:1015--1021, 2006.

\bibitem{GameBook}
E.~Gr\"adel, W.~Thomas, and T.~Wilke.
\newblock {\em Automata, Logic and Infinite Games}, volume LNCS 2500.
\newblock Springer, 2002.

\bibitem{gugercin}
S.~Gugercin.
\newblock {\em Projection methods for model reduction of large-scale dynamical
  systems}.
\newblock PhD thesis, Rice Univ., Houston, TX, May 2003.

\bibitem{Habets1}
C.G.J.M. Habets and J.~H. van Schuppen.
\newblock Reduction of affine systems on polytopes.
\newblock In {\em International Symposium on Mathematical Theory of Networks
  and Systems}, 2002.

\bibitem{Hazewinkel1}
M.~Hazewinkel.
\newblock Moduli and canonical forms for linear dynamical systems {II}: The
  topological case.
\newblock {\em Mathematical Systems Theory}, 10:363--385, 1977.

\bibitem{Kotsalis2}
G.~Kotsalis, A.~Megretski, and M.~A. Dahleh.
\newblock Balanced truncation for a class of stochastic jump linear systems and
  model reduction of hidden {M}arkov models.
\newblock {\em IEEE Transactions on Automatic Control}, 53(11), 2008.

\bibitem{Kotsalis1}
G.~Kotsalis and A.~Rantzer.
\newblock Balanced truncation for discrete-time {M}arkov jump linear systems.
\newblock {\em IEEE Transactions on Automatic Control}, 55(11), 2010.

\bibitem{VardKupfContr}
O.~Kupferman, P.~Madhusudan, and P.S. Thiagarajan.
\newblock Open systems in reactive environments: Control and synthesis.
\newblock In {\em CONCUR'00}, 2000.

\bibitem{liberzon2003}
D.~Liberzon.
\newblock {\em Switching in Systems and Control}.
\newblock Birkh{\"a}user, Boston, MA, 2003.

\bibitem{BilinearMomentMatching2}
Y.~Lin, L.~Bao, and Y.~Wei.
\newblock A model-order reduction method based on krylov subspaces for mimo
  bilinear dynamical systems.
\newblock {\em Journal of Applied Mathematics and Computing}, 25(1-2):293--304,
  2007.

\bibitem{ljung}
L.~Ljung.
\newblock {\em System Identification, Theory for the User}.
\newblock Prentice Hall, Englewood Cliffs, NJ, 1999.

\bibitem{Mazzi1}
E.~Mazzi, A.S. Vincentelli, A.~Balluchi, and A.~Bicchi.
\newblock Hybrid system model reduction.
\newblock In {\em IEEE International conference on Decision and Control}, 2008.

\bibitem{6209392}
N.~Monshizadeh, H.~Trentelman, and M.~Camlibel.
\newblock A simultaneous balanced truncation approach to model reduction of
  switched linear systems.
\newblock {\em Automatic Control, IEEE Transactions on}, PP(99):1, 2012.

\bibitem{MP:BigArticlePartI}
M.~Petreczky.
\newblock Realization theory for linear and bilinear switched systems: formal
  power series approach - part i: realization theory of linear switched
  systems.
\newblock {\em ESAIM Control, Optimization and Caluculus of Variations},
  17:410--445, 2011.

\bibitem{MP:HybIoDADHS09}
M.~Petreczky, P.~Collins, D.A. van Beek, J.H. van Schuppen, and J.E. Rooda.
\newblock Sampled-data control of hybrid systems with discrete inputs and
  outputs.
\newblock In {\em Proceedings of 3rd IFAC Conference on Analysis and Design of
  Hybrid Systems (ADHS09)}, 2009.

\bibitem{PM12}
M.~Petreczky and G.~Merc{\`e}re.
\newblock Affine {LPV} systems: realization theory, input-output equations and
  relationship with linear switched systems.
\newblock In {\em Proceedings of the {IEEE} Conference on Decision and
  Control}, Maui, Hawaii, USA, December 2012.

\bibitem{petreczkypeeters1}
M.~Petreczky and R.~Peeters.
\newblock Spaces of nonlinear and hybrid systems representable by recognizable
  formal power series.
\newblock In {\em Proc. 19th International Symposium on Mathematical Theory of
  Networks and Systems}, pages 1051--1058, Budapest, Hungary, July 2010.

\bibitem{petreczky}
M.~Petreczky and J.~H. van Schuppen.
\newblock Partial-realization theory for linear switched systems - a formal
  power series approach.
\newblock {\em Automatica}, 47:2177–--2184, October 2011.

\bibitem{petreczky2013}
M.~Petreczky, R.~Wisniewski, and J.~Leth.
\newblock Balanced truncation for linear switched systems.
\newblock {\em Nonlinear Analysis: Hybrid Systems}, 10:4--20, November 2013.

\bibitem{shaker2011}
H.~Shaker and R.~Wisniewski.
\newblock Generalized gramian framework for model/controller order reduction of
  switched systems.
\newblock {\em International Journal of System Science}, 42:1277–--1291,
  August 2011.

\bibitem{Sun:Book}
Z.~Sun and S.~S. Ge.
\newblock {\em Switched linear systems : control and design}.
\newblock Springer, London, 2005.

\bibitem{TabuadaBook}
P.~Tabuada.
\newblock {\em Verification and Control of Hybrid Systems: A Symbolic
  Approach}.
\newblock Springer-Verlag, 2009.

\bibitem{toth2012}
R.~T\'{o}th, H.~S. Abbas, and H.~Werner.
\newblock On the state-space realization of lpv input-output models: Practical
  approaches.
\newblock {\em {IEEE} Trans. Contr. Syst. Technol.}, 20:139--153, January 2012.

\bibitem{Wonham3}
W.M. Wonham.
\newblock Supervisory control of discrete-event systems.
\newblock Lecture notes, \verb'http://www.control.utoronto.ca/~wonham'.

\bibitem{China2}
L.~Zhang, E.~Boukas, and P.~Shi.
\newblock Mu-dependent model reduction for uncertain discrete-time switched
  linear systems with average dwell time.
\newblock {\em International Journal of Control}, 82(2):378-- 388, 2009.

\bibitem{China3}
L.~Zhang and P.~Shi.
\newblock Model reduction for switched lpv systems with average dwell time.
\newblock {\em IEEE Transactions on Automatic Control}, 53:2443--2448, 2008.

\bibitem{Zhang20082944}
L.~Zhang, P.~Shi, E.~Boukas, and C.~Wang.
\newblock Model reduction for uncertain switched linear discrete-time systems.
\newblock {\em Automatica}, 44(11):2944 -- 2949, 2008.

\end{thebibliography}

%
%

%

%
%
%




\end{document}